\documentclass[twoside, reqno]{amsart}
\usepackage{amssymb}
\usepackage{amsfonts}
\usepackage{graphicx}
\usepackage{psfrag}
\usepackage{float}
\usepackage{subcaption}
\usepackage{amsmath}
\usepackage{amsthm}
\usepackage{physics}
\usepackage{enumitem}
\usepackage{ mathrsfs}
\usepackage[dvipsnames]{xcolor}
\usepackage[pagewise]{lineno}

\usepackage[left=2.5cm, right=2.5cm, top=3cm, bottom=2.5cm]{geometry}
\usepackage[colorlinks=true, pdfstartview=FitV, linkcolor=blue,citecolor=blue, urlcolor=blue]{hyperref}
\usepackage{cleveref}

\usepackage{xcolor}
\definecolor{labelkey}{rgb}{0,0,1}

\usepackage{amsfonts, amssymb, amsmath, amsthm, mathrsfs, bbm, cjhebrew, gensymb, textcomp, mathtools, dsfont,tikz}
\usepackage[normalem]{ulem}
\usepackage{subcaption}

\newtheorem{remark}{Remark}
\newtheorem{lemma}{Lemma}
\numberwithin{lemma}{section}

\newtheorem{theorem}{Theorem}
\numberwithin{theorem}{section}

\newtheorem{proposition}{Proposition}
\numberwithin{proposition}{section}

\numberwithin{equation}{section} 

\makeatletter

\def\Tr@nsmogrify#1#2.{\expandafter\newcommand\csname #1#2\endcsname
{\mathchardef\Tr@ns@temp=\mathcode\lccode`#1\relax
\mathcode\lccode`#1=\mathcode`#1\lowercase{\csname#1#2\endcsname}%
\mathcode\lccode`#1=\Tr@ns@temp\relax}}

\@for\x:=Sin,Cos,Max,Min,Lim,Limsup,Liminf,Inf\do{%
\expandafter\Tr@nsmogrify \x.}

\makeatother

\newcommand{\lpj}{\triangle_j}
\newcommand{\lpk}{\triangle_k}

\newcommand{\ptl}{\partial}
\newcommand{\te}{\theta}
\newcommand{\ga}{\gamma}
\newcommand{\la}{\Lambda}
\newcommand{\ka}{\kappa}

\newcommand{\zz}{\mathbb{Z}^2}
\newcommand{\C}{\mathbb{C}}

\newcommand{\E}{\mathfrak{E}}

\newcommand{\Hdot}{\dot{H}}
\newcommand{\Acal}{\mathcal{A}}

\newcommand{\kbar}{\overline{\ka}}

\newcommand{\lbn}{\left (}
\newcommand{\rbn}{\right )}

\newcommand{\bkap}{\bar{\kap}}
\newcommand{\Sob}[2]{\lVert#1\rVert_{#2}}
\newcommand{\nrm}[1]{\lVert#1\rVert}

\newcommand{\kapb}{\underline{\kap}}
\newcommand{\kapd}{\kap_{\textrm{d}}}
\newcommand{\Gb}{\underline{G}}

\newcommand{\RR}{\mathbb{R}}
\newcommand{\TT}{\mathbb{T}}
\newcommand{\ZZ}{\mathbb{Z}}

\newcommand{\Bcal}{\mathcal{B}}
\newcommand{\Ecal}{\mathcal{E}}

\newcommand{\al}{\alpha}
\newcommand{\be}{\beta}
\newcommand{\de}{\delta}
\newcommand{\De}{\Delta}

\newcommand{\gam}{\gamma}

\newcommand{\kap}{\kappa}

\newcommand{\Lam}{\Lambda}

\newcommand{\om}{\omega}
\newcommand{\Om}{\Omega}
\newcommand{\si}{\sigma}

\newcommand{\tht}{\theta}

\newcommand{\bv}{\mathbf{v}}
\newcommand{\bw}{\mathbf{w}}

\newcommand{\bdy}{\partial}
\newcommand{\lb}{\langle}
\newcommand{\rb}{\rangle}

\DeclareMathOperator{\supp}{supp}

\begin{document}

\title{gSQG}
\date{October 16, 2025}
\title[gSQG turbulence]
{On Turbulent behavior of the generalized surface quasigeostrophic equations}
\author{Chengzhang Fu$^{1}$}
\author{Michael S. Jolly$^{1,\dagger}$}
\address{$^1$ Department of Mathematics, Indiana University, Bloomington, IN}
\author{Anuj Kumar$^{2}$}
\address{$^{2}$ New Delhi, India}
\author{Vincent R. Martinez$^{3,4}$}
\address{$^{3}$ Department of Mathematics and Statistics,
CUNY Hunter College, New York City, USA}
\address{$^{4}$ Department of Mathematics,
CUNY Graduate Center, New York City, USA}

\email[C. Fu]{fu6@iu.edu}
\email[M. S. Jolly]{msjolly@iu.edu}
\email[A Kumar]{akumar241@outlook.com}
\email[V. R. Martinez]{vrmartinez@hunter.cuny.edu}

\thanks{M.S.J. was supported by Simons grant MP-TSM-00002337; V.R.M. was in part supported by the National Science Foundation through DMS 2213363, DMS 2206491, DMS 2511403, the Simons Foundation through MP-TSM 00014320, and the Dolciani Halloran Foundation. The computational component was supported in part by Lilly Endowment, Inc., through its support for the Indiana University Pervasive Technology Institute.}

\date{October 16, 2025}


\begin{abstract}
Turbulent behavior of the two-parameter family of generalized surface quasigeostrophic equations is examined both rigorously and numerically.  We adapt a cascade mechanism argument to derive an energy spectrum that scales as $\ka^{2\beta/3-3}$ where $\beta$ controls the regularity of the velocity ($\beta=1$ in the special case of the SQG). Direct numerical simulations indicate that this fits better than $\ka^{\beta/3-3}$ which {was} derived in \cite{Pierrehumbert1994}. Guided by earlier work on the 2D Navier--Stokes equations, we prove a certain condition implies  a direct cascade of enstrophy, as well as an upper bound on the enstrophy dissipation rate, and sharp bounds on a dissipation wavenumber. The dependence of these rigorous results on the two parameters is demonstrated numerically.
\end{abstract}
\maketitle

\section{Introduction}
 We study the turbulent behavior of the generalized surface quasi-geostrophic equation (gSQG) over the domain $\Omega=[0,L]^2$. The gSQG equation is given by
\begin{align} \label{maineq}
    &\ptl_t \te + \ga \la ^{\alpha} \te + u \cdot \nabla \te = g,\quad u=\nabla^{\perp}\psi \overset{\text{def}}{=}(-\bdy_{x_2} \psi, \bdy_{x_1} \psi),\quad
    \Delta \psi = \ka_0^{-\beta}\la^{\beta} \te,
\end{align}
where $0\leq \beta < 2$, $0<\alpha \leq 2$, and $\ka_0=2\pi/L$. {The driving force $g$ is time-independent and given; this acts as a large-scale energy source to sustain turbulent behavior. The damping parameter, $\gam$, is positive, so the fractional laplacian, $\la^\al$, acts as a small-scale energy sink; recall that $\la\overset{\text{def}}{=}(-\Delta)^{\frac{1}{2}}$. Lastly, we equip \eqref{maineq} with periodic boundary conditions and assume that $\tht$, $g$ are mean-free over $\Om$.  Note that the factor $\ka_0^{-\beta}$ is included in the constitutive law so that $u$ retains the dimensions of velocity and the overall dimensional consistency of the equation is maintained.}   When $ \alpha = 2 $ and $ \beta = 0 $, the equation reduces to the vorticity formulation of the Navier--Stokes equations (NSE), while fixing $ \beta = 1 $ yields the special case of the surface quasi-geostrophic equation (SQG).

The theories for two-dimensional turbulence by Batchelor, Kraichnan, and Leith
\cite{B53, Kraichnan1967, Leith} are akin to that for 3D by Kolmogorov \cite{KO1991} in that they
are derived through scaling arguments without direct use of the equations of motion.  The main difference is that in 2D there are two invariances, one for energy and one for enstrophy, which enable cascades of both quantities toward larger and smaller scales, respectively.  Foias \cite{foias1997turbulence} made explicit use of the Navier-Stokes equations to provide rigorous support for certain elements of these theories.  That approach was continued in \cite{FJMR2002,Dascaliuc2008}, which contain arguments that are modified here for  
some of our results for
the gSQG.

The theoretical study of the inviscid gSQG family was introduced to the mathematical community in \cite{ChaeConstantinCordobaGancedoWu2012}, where the weak solution theory, local strong solution theory, and local theory for the corresponding patch problem were initially developed. Since then, several works have subsequently refined the understanding of when the initial value problem corresponding to \eqref{maineq} and associated modifications are well-posed \cite{ChaeWu2012, ChaeConstantinWu2012b, HuKukavicaZiane2015, KiselevYaoZlatos2017, YuZhenJiu2019, GancedoPatel2021, JKM2022, ChaeJeongOh2023b} or ill-posed \cite{KiselevRyzhikYaoZlatos2016, KukavicaVicolWang2016, JeongKim2021, ChaeJeongOh2023b, CordobaLucas-ManchonZoroa-Martinez2025, ChoiJungKim2025} in various contexts, as well as the construction of global solutions \cite{NahmodPavlovicStaffilaniTotz2018, Nguyen2018, CordobaGomez-SerranoIonescu2019, GeldhauserRomito2020, Rosenzweig2020, CaoQinZhanZou2023}. The global issue of whether solutions emanating from smooth initial data blow-up in finite time remains an outstanding open problem. 
On the other hand, in the presence of fractional dissipation, global regularity has been resolved within a subset of the family \cite{Resnick1995, ConstantinWu1999, ConstantinGlattHoltzVicol2013}, most notably in the presence of so-called \textit{critical dissipation}, where the dissipation power is related to the constitutive law in a particular way \cite{KiselevNazarovVolberg2007, KiselevNazarov2009, CaffarelliVasseur2010, ConstantinVicol2012, ConstantinIyerWu2008, LazarXue2019, MiaoXue2012}. Generally speaking, the dissipative SQG family whose turbulent behavior we study in this article is known to be locally well-posed for large initial data and globally well-posed for small initial data \cite{MiaoXue2011, ChaeConstantinWu2012a, JollyKumarMartinez2020a}.

The particular features 
of turbulence
considered in this paper are the energy spectrum, a direct cascade of enstrophy, and a 2D analogue of Kolmogorov's dissipation law.  There are two wavenumbers that play critical roles: $\ka_\eta$ where the spectrum is expected to start decaying exponentially and $\ka_\sigma$, a Dirichlet quotient which can determine the extent of the cascade range.    The relevant physical quantities are rescaled from the NSE case in terms of $\alpha$ and $\beta$ which also appear in the relations between them.  
We first use the Richardson/Kraichnan cascade mechanism to derive a power law for the energy spectrum reflecting this rescaling. That is followed by an estimate which guarantees a pronounced direct cascade of energy from the forcing scale to a fraction of $\ka_\sigma$.  We then show rigorously that if the power law holds, other features of turbulence follow.  The first is that $\ka_\sigma$ is comparable to $\ka_\eta$, up to a factor that scales as $\beta^{-1}$.  The second
is that both wave numbers should scale as
$G^{\frac{1}{2\alpha}}$ in the Grashof number $G$.

Two results are proved without assuming the spectrum.  One is the general bound $G^{\frac{1}{3\alpha}} \lesssim \ka_\eta/ \ka_0 \lesssim G^{\frac{2}{3\alpha}}$. The other is one side of the dissipation law, namely $\eta \lesssim U^3/L^3$, where $\eta$ is the enstrophy dissipation rate and $U$ is a suitably scaled quantity that reduces to the root mean square velocity in the case of the NSE.  
We note that while most of the proofs we provide are rescaled adaptations of results for the NSE case in \cite{FJMR2002, Dascaliuc2008}. However, the proofs of dissipation law $\eta \lesssim U^3/L^3$, as well as the tighter bounds on $\ka_\eta$, $\ka_\sigma$ involve a more delicate estimate of the nonlinear term involving commutators. Some technical background for this approach is included in the Appendix. 

Although there is no known rigorous derivation of the power law for the spectrum, it has been observed in countless numerical studies of the NSE.   We present here high-resolution simulations (up to $32,768^2$ collocation points) to test the rescaled power law for the gSQG.  We note where in the $\alpha,\beta$-plane the spectrum for gSQG starts to deviate from the heuristically predicted power law and in turn the extent to which the relations that would follow from that law fail.  We consistently find a marked breakdown when we cross a critical line where the gSQG changes from being quasilinear to fully nonlinear. 

\section{Mathematical Preliminaries}

We identify the domain $\Om$ with the two-dimensional torus:
    \begin{align}\notag
        \Omega=[0,L]^2=\TT^2
    \end{align}
and in the analysis retain the factor $\kap_0$ to help track the dependence of physical dimensions line-to-line. 
 We denote the phase space of \eqref{maineq} by $H$, which we define to be the subspace of $L^2(\TT^2)$ of real-valued, scalar functions which are mean-zero over $\TT^2$. Within this space-periodic setting, we make use the {standard Fourier series framework}. Thus, $H$ can be characterized as
\begin{equation*}
    H := \left\{ \te = \sum_{k\in \zz} \hat{\te}_k e^{i{\kap_0}k\cdot x} \in L^2(\TT^2): \hat{\te}_k \in \C,\  \hat{\te}_{\vec{0}} = 0,\ \hat{\te}_{-k} = \overline{\hat{\te}_k} \right\},
\end{equation*}
where $\hat{\te}_k$ denotes the Fourier coefficient of $\tht$ at wavenumber $k\in\ZZ^2$. The scalar product in $H$ is simply the $L^2$-inner product and we will denote by
\begin{equation*}
    (\te_1,\te_2) = \int_{\TT^2} \te_1(x)\te_2(x)dx.
\end{equation*}
We will denote the associated norm in $H$ by
\begin{equation*}
    |\te| = (\te,\te)^{\frac{1}{2}} = \left(\int_{\TT^2} \te^2dx\right)^{\frac{1}{2}} .
\end{equation*}
Parseval's identity can be read as
\begin{equation*}
   {|\te|}^2 = 4\pi^2 \sum_{k \in \zz} |{\hat{\te}_{k}}|^2 \text{ as well as } (\te,\te')=4\pi^2 \sum_{k \in \zz} {\hat{\te}_{k}} \cdot {\hat{\te}_{-k}'}
\end{equation*}
for $ \te'  =\sum_{k\in \zz} {\hat{\te}_{k}'}e^{ik\cdot x}$. 

The fractional laplacian operator $\la=(-\Delta)^{\frac{1}{2}}$ is self-adjoint and can be defined spectrally. Its eigenvalues are of the form ${\kap_0}|k|$ where $ k \in \zz \backslash \{0\} $; the eigenvalues are denoted and arranged as $0<\lambda_0 = 1 \leq \lambda_1 \leq \lambda_2 \leq ...$, where they are counted according to their multiplicities. Let $w_0,w_1,w_2,...$ be the corresponding normalized eigenvectors, i.e., $|w_i|=1$, for all $j$. {Then for each $\tht\in H$, we have}
    \begin{align*}
        {\tht(x)=\sum_{k\in\ZZ^2}\hat{\tht}_ke^{i\kap_0k\cdotp x}=\sum_{j=0}^\infty (\tht,w_j)w_j(x).}
    \end{align*}
For $\sigma \ge 0$, the {positive powers} of $\la$ are defined by linearity through
\begin{equation*}
    \la^{\sigma}w_j = \lambda_j^{\frac{\sigma}{2}}w_j, \ \ \textrm{for} \ \ j=0,1,2,...
\end{equation*}

We define projectors $P_{\kappa}: H \to \text{span}\{w_j : \lambda_j \leq \kappa\}$ by
\begin{equation} \label{proj:p}
    P_{\kappa}\te = \sum_{|k| \leq \kappa} \hat{\te}_ke^{ik\cdot x}
\end{equation}
with $Q_{\kappa}=I-P_{\kappa}$. {In our analysis, it will be useful to consider components of $\te$ within a range in wavenumbers, so we define}
\begin{equation*}
    \te_{\ka,\ka'} = (P_{\kappa'} - P_{\kappa})\te
\end{equation*}
for $0\leq \ka < \ka'$, with the convention that $\te_{\ka,\infty}=\te_{\ka}={Q_\ka\te}$ for all $0\leq \ka < 1$. {Note that}
    \begin{align}
        \tht_{\kap,\kap'}=Q_\kap P_{\kap'}\tht=P_{\kap'}Q_\kap\tht.\notag
    \end{align}


{Lastly, recalling that $u=\nabla^\perp\psi$, and 
is
thus divergence-free, we note that} the nonlinear term satisfies orthogonality relations similar to ones known for the NSE. In particular, for $u,\tht, w$ sufficiently smooth, one has
    \begin{align}\label{eq:identities}
        (u\cdotp\nabla\te,w) = -(u\cdotp\nabla w,\te)
    \end{align}
and hence
    \begin{align}\label{bterm_prod}
     (u\cdotp\nabla\te, \te)=0\;.
    \end{align}
Moreover, from the vector identity $\bv^\perp\cdot \bw=-\bw^\perp\cdot\bv$ and \eqref{eq:identities} one has
\begin{equation} \label{psi_orthog}
    (u\cdot\nabla \tht, \psi)=(\nabla^\perp \psi\cdot \nabla \tht, \psi)=-(\nabla^\perp \psi\cdot \nabla\psi, \tht) =0.
\end{equation}

\subsection{Apriori Estimates}

{Proceeding formally, if} we multiply \eqref{maineq} by $-\psi$ (respectively, $\te$), then integrate over $\TT^2$ and apply \eqref{bterm_prod}, we find that
    \begin{align} 
        \frac{1}{2}\frac{d}{dt}|\la^{\frac{\beta-2}{2}}\te|^2 + \ga |\la^{\frac{\alpha+\beta-2}{2}}\te|^2 &= (g,-\psi)\label{def_energy}
        \\
        \frac{1}{2}\frac{d}{dt}|\te|^2 + \ga |\la^{\frac{\alpha}{2}}\te|^2 &= (g,\te)\;. \label{def_enstrophy}
    \end{align}
We define 
\begin{equation*}
    \frac{1}{L^2}|\la^{\frac{\beta-2}{2}}\te|^2 \overset{\text{def}}{=} \text{2 times the total ``energy" per unit mass}
\end{equation*}
and
\begin{equation*}
    \frac{1}{L^2}|\te|^2 \overset{\text{def}}{=} {2\ \text{times}\ } \text{the total ``enstrophy" per unit mass} .
\end{equation*}
Note that {in the context of the NSE}, when $\beta=0$, these match the {conserved quantities of energy, $\frac{1}{L^2}|u|^2$, and enstrophy, $\frac{1}{L^2}|\om|^2$, respectively, (per unit mass)} {where $\te$ is interpreted as the fluid vorticity $\om=\nabla^\perp\cdotp u$}.

The relations \eqref{def_energy} and \eqref{def_enstrophy} are the balance equations for the energy and enstrophy, respectively. Applying the Poincar\'e, Cauchy-Schwarz and Young inequalities to \eqref{def_enstrophy}, we find that
\begin{equation}\label{prod_G_a} 
  \frac{d}{dt} |\te|^2 + \kappa_0^{\alpha} \ga |\te|^2 \leq \frac{d}{dt} |\te|^2 + \ga |\la^{\frac{\alpha}{2}}\te|^2 \leq \frac{|\la^{-\frac{\alpha}{2}}g|^2}{\ga} ,    
\end{equation}
so that the Gronwall lemma gives
\begin{equation}\label{prod_G}
    {\sup_{t\geq t_*}|\te(t)|^2} \le 2\frac{|\la^{-\frac{\alpha}{2}}g|^2}{\ga^2\kappa_0^{\alpha}}
\end{equation}
for $t_*$ sufficiently large, {depending on $\gam,\kap_0,\al, |\theta_0|, |\Lam^{-\frac{\al}2}g|$.} 

 {Arguing similarly, we derive}
\begin{equation*}  
     \frac{d}{dt} |\la^{\frac{\beta-2}{2}}\te|^2 + \ga |\la^{\frac{\alpha+\beta-2}{2}}\te|^2 \leq \frac{|\la^{\frac{\beta-\alpha-2}{2}}g|^2}{\ga},
\end{equation*}
{from which we deduce
\begin{align}\label{prod_Gstar}
 \sup_{t\geq t_*}|\la^{\frac{\beta-2}{2}}\te(t)|^2 \le 2\frac{|\la^{\frac{\beta-\alpha-2}{2}}g|^2}{\ga^2\kappa_0^{\alpha}} 
\end{align}
for sufficiently large $t_*$, depending on $\gam,\kap_0,\al, |\Lam^{\frac{\be-2}2}\theta_0|, |\Lam^{\frac{\be-\al-2}2}g|$. }

The inequalities in \eqref{prod_G} and \eqref{prod_Gstar} can {equivalently} be expressed in terms of the dimensionless {\it Grashof numbers} $G$, ${G_*}$ as 
 $$|\te| \leq 2\ga\kappa_0^{\alpha-1}G \quad \text{and}\quad |\la^{\frac{\beta-2}{2}}\te| \leq {2}\ga \kappa_0^{\frac{2\alpha+\beta-4}{2}} G_*$$ where
    \begin{equation} \label{def:G}
        G \overset{\text{def}}{=} \frac{|\la^{-\frac{\alpha}{2}}g|}{\ga^2 \kappa_0^{\frac{3\alpha-2}{2}}} \qquad \text{and} \qquad G_* \overset{\text{def}}{=} \frac{|\la^{\frac{\beta-\alpha-2}{2}}g|}{\ga^2 \kappa_0^{\frac{3\alpha+\beta-4}{2}}} .
    \end{equation}
{We note that}
\begin{equation}\label{eq:G:Gstar:equiv}
    \left(\frac{\kapb}{\kappa_0}\right)^{\frac{2-\beta}{2}}  G_*\leq G \leq \left(\frac{\bkap}{\kappa_0}\right)^{\frac{2-\beta}{2}} G_* 
\end{equation}
whenever $0<\kapb<\bkap$ and $g$ satisfies 
\begin{equation} \label{def:g_kbound}
    g=\sum_{{\kapb} < |k| \leq \bkap} \hat{g}_k e^{ik\cdot x}=g_{\kapb,\bkap}.
\end{equation} 
{In other words, \eqref{eq:G:Gstar:equiv} holds whenever the external driving force $g$ has finite spectral support.}

\subsection{Some Remarks on the Mathematical Framework}
It is common in the physics and engineering literature to assume for turbulent flows that the time averages of physically relevant quantities exist and are independent of the initial condition.   In the sequel, we adopt this view and define 
$$
\langle \Phi(\theta) \rangle = \lim_{t\to \infty} \frac{1}{t} \int_0^t \Phi (S(\tau)\theta_0) \ d\tau\;.
$$
While it is conceivable that this limit would not exist, it can be replaced with a generalized limit, denoted $\Lim$, that is guaranteed to exist by the Hahn-Banach theorem and matches the ordinary limit when {it does exist}.  Moreover, the $\Lim$ {functional} can typically be expressed as an integral with respect to an invariant measure {supported on the global attractor of the system provided that the global attractor exists}. {Thus, in contexts where a global attractor theory is available}, these technical adjustments make averaging over solutions {mathematically} rigorous. 

{In the particular case of the 2D NSE, such a framework exists and a systematic study of the Kraichnan theory from a first principles perspective can indeed be developed. We refer the reader to \cite{FJMR2002} for such a study and additional details regarding this rigorous framework. We do not address these concerns in the paper, although we remark that global attractors for certain subsets of the gSQG family of equations have been established, namely the critical and subcritical regimes of dissipative SQG, i.e., $\be=1$, $\al\in[1,2]$, in \cite{NJ2005, ConstantinTarfuleaVicol2015}. To our best knowledge, it remains an open direction to develop the global attractor theory for the gSQG family in general.}

{We conclude these remarks by emphasizing that, for our purposes, the absence of such a framework does not alter our formal calculations. {In fact, we point out that any theory of weak solutions for which global-in-time existence can be guaranteed from arbitrary initial data in $H$ and for which the energy balance \eqref{def_energy} and enstrophy balance \eqref{def_enstrophy} hold with equality is sufficient to justify the analysis performed in paper. We will not pursue the development of this solution theory here and instead focus on the consequences for turbulence that emanate from these putative solutions. Lastly, regarding our numerical results, the time-averages} that are calculated in our experiments are computed by simply taking $t$ large. The inherent error due 
to {implementing} finite time averages {can then be} estimated in terms of the Grashof number {(see \cite{FJM2005} for details in the case of 2D NSE)}.

}

\section{Energy spectrum}
\subsection{Heuristic derivation}
The energy spectrum, which describes how energy varies over different length scales,  plays a central role in turbulence.  Although it is widely observed in experiments and numerical simulations to satisfy a universal power law, a rigorous derivation, which would have to depend on the nature of the force, is not known even in the case of the NSE.   However, there are several heuristic arguments for it.   We carry out one adapted from Kraichnan's 2D interpretation of Richardson's cascade mechanism \cite{Kraichnan1967}.   To start, based on \eqref{def_enstrophy}, we define the \textit{total {enstrophy} dissipation rate per unit mass} as
\begin{equation} \label{def:eta}
    \eta = \frac{\ga}{L^2} \langle|\la^{\frac{\alpha}{2}} \te|^2\rangle \;.
\end{equation}

Let
\begin{equation*}
    e_{\kappa} = \text{2 times the average energy per unit mass of the eddies of linear size } \mathit{l} \in \left[\frac{1}{2\ka}, \frac{1}{\ka} \right) \;.
\end{equation*}
In terms of the solution of the gSQG equation, the quantity $e_{\kappa}$ can be expressed as  
\begin{equation} \label{def:e_k}
    e_{\kappa} = {\frac{1}{L^2}}\big\langle \,|\Lambda^{\tfrac{\beta - 2}{2}} \theta_{\kappa,2\kappa}|^2 \,\big\rangle .
\end{equation}
An analogue of the average velocity of eddies of size $l$ is defined by  
\begin{equation*}
    U_\kappa = {L^{\frac{\be}2}}
    e_\kappa^{1/2}
      = {L^{\frac{\be}2-1}}
        \big\langle \,|\Lambda^{\tfrac{\beta - 2}{2}} \theta_{\kappa,2\kappa}|^2 \,\big\rangle^{1/2}.
\end{equation*}
Correspondingly, the average time for these eddies to travel a distance of order $l$ is given by  
\begin{equation*}
    t_\kappa = \frac{l}{U_\kappa} \sim \frac{1}{\kappa U_\kappa}
      \sim
      \frac{\kappa_0^{\beta/2}}{\kappa e_\kappa^{1/2}}.
\end{equation*}
The average enstrophy per unit mass associated with scale $l$ is defined as
\[
E_\kappa = \kappa^{2-\beta} e_\kappa,
\]
which yields the enstrophy dissipation rate for eddies of length $l$:  
\begin{equation*}
    \eta_\kappa \sim \frac{E_\kappa}{t_\kappa}
      \sim \kappa_0^{-\beta/2}\,\kappa^{3-\beta}\, e_\kappa^{3/2}.
\end{equation*}
{In particular, this implies that}
    \begin{align}\notag
        e_\kappa \sim \eta_\kappa^{2/3}\,\kappa_0^{\beta/3}\,\kappa^{\tfrac{2\beta}{3}-2}
    \end{align}
{We recall that within the inertial range of wavenumbers, it is expected that the relation $\eta_\kappa \approx \eta$ holds. Upon assuming that $\eta_\kappa\approx\eta$ holds, one may deduce}
\begin{equation} \label{def:eta_k}
    e_\kappa \sim \eta^{2/3}\, \kappa_0^{\beta/3}\,\kappa^{\tfrac{2\beta}{3}-2}.
\end{equation}
{With this observation in hand, we will now use \eqref{def:eta_k} to \textit{define} the ``inertial range." In particular, we define {\it inertial range of wavenumbers} to be the interval of wavenumbers of amplitude $\kappa$ over which \eqref{def:eta_k} holds.  Note that this power law depends explicitly on $\beta$, but it depends on $\alpha$ only implicitly through $\eta$.}

The time average
\begin{equation*}
    \lim_{t \to \infty} \frac{1}{t} \int_0^t \frac{1}{L^2} |\la^{\sigma} (P_{\ka'}-P_{\ka}) \la^{\frac{\beta-2}{2}}S(\tau)\te_0|^2d\tau .
\end{equation*}
can be written as a Riemann sum (as $L\to \infty$, or equivalently $\ka_0\to 0$) for an integral in the wavenumbers
\begin{equation}\label{eq:spectrum:riemann}
    \int_{\ka}^{\ka'} \chi^{2\sigma} \mathcal{E} (\chi)\ d\chi
\end{equation}
of some function $\mathcal{E}$, called the {\it energy spectrum} of the turbulent flow sustained by the force $g$.  \
The spectrum is related to the average energy per unit mass $e_{\kappa}$ through
\begin{equation*}
     \int_{\ka}^{2\ka} \mathcal{E} (\chi) d\chi \sim e_{\ka}  .
\end{equation*}
{Thus, by \eqref{def:eta_k},
it is} expected to satisfy
    \begin{equation} \label{main:ES}
    \mathcal{E}(\kappa) \sim \eta^{2/3}\ka^{2\beta/{3}-3} .
    \end{equation}
in the inertial range.  
  
Let ${\kapd}$ denote {the wavenumber cut-off} where inertial effects {achieve a sustained balance with small-scale} viscous effects.  While {a} precise expression for this wavenumber is not known, {one may} again use dimensional analysis to {establish a putative relationship} to known quantities.  This is done by assuming {that $\kap_{\text{d}}$} depends on only  $ \gamma $ and $ \eta $ though some function $\varphi$:
\[
{\kapd} = \varphi(\gamma, \eta).
\]
We consider the following rescaling of $\kap_d'$:
\begin{equation*}
    {\kapd'} = \varphi(\ga',\eta'), \quad {\kapd} = \frac{{\kapd'}}{\xi},\quad  \ga' = \frac{\xi^{\alpha}}{\tau}\ga, \quad \eta' = \frac{1}{\tau^3} \eta.
\end{equation*}
It follows that
\begin{equation*}
    \frac{1}{\xi}\varphi(\ga',\eta') = \varphi(\xi^{\alpha}\tau^{-1}\ga,\frac{1}{\tau^3} \eta)\;.
\end{equation*}
Thus, if we choose $\frac{\xi^{\alpha}}{\tau} = \frac{1}{\ga}$ and $\tau^3 = \eta$, we obtain
\begin{equation} \label{def:keta}
    {\kapd} = \varphi(\ga,\eta) = \xi \varphi(1,1) \sim \left(\frac{\eta}{\ga^3} \right)^{1/(3\alpha)} \overset{\text{def}}{=} \ka_\eta \;.
\end{equation}
In contrast to the energy spectrum power law, $\kappa_\eta$ depends explicitly on  $\alpha$, {while the} dependence on $\beta$ {remains} implicit through $\eta$.

If we assume that 
$\mathcal{E}(\kappa) \sim \varphi(\eta, \ka)$
{holds} in an inertial range of { wavenumbers}, for some function $\varphi$, a similar dimensional argument results in the spectrum
    \begin{equation} \label{main:ES2}
    {\mathcal{E}_2}(\kappa) \sim \eta^{2/3}\ka^{\beta-3} \;.
    \end{equation}
Yet another dimensional analysis by Pierrehumbert, Held, and Swanson \cite{Pierrehumbert1994} based on locality in scale results in 
    \begin{equation} \label{main:ES3}
    {\Ecal_{PHS}}(\kappa) \sim \eta^{2/3}\ka^{\beta/3-3} \;.
    \end{equation}
All three {energy spectra, \eqref{main:ES}, \eqref{main:ES2}, \eqref{main:ES3}, are consistent} with Kraichnan's $\kappa^{-3}$ spectrum
in the {special case $\be=0$ corresponding to} the 2D NSE. 

\subsection{Numerical simulation setup}
All computations are done with a fully dealiased pseudospectral code with $N$ modes in each direction.  For the NSE we use $N=16384$ and for the gSQG $N=32768$.  The force is restricted to  $\kappa \in [9,12]$  as in \eqref{def:g_kbound}. The total number of forced modes is 94, which is half the number of lattice points in the annulus with inner and outer radius ${\kapb}$ and $\bkap$. Given that the force is real so that $\hat{g}_k = -\hat{g}_{-k}$, we only need to consider half of them independently.
We randomly assign Fourier coefficients to these modes, choosing values from $(-1,1)$ for both the real and imaginary parts. After selecting the coefficients, we multiply each by a factor of $10^{-5}$ to enable computations for smaller viscosity values.

In the NSE case, we calculate that the selected forcing satisfies $ |g| = 5.1354 \times 10^{-4} $ and $ |\Lambda^{-1} g| = 4.935 \times 10^{-5} $. We take the viscosity within the range $ 10^{-9} $ to $ 2 \times 10^{-7} $, which corresponds to a Grashof numbers in the range $ 1.2338 \times 10^{8} $ to $ 4.935 \times 10^{13} $. We show that this is sufficiently large to produce turbulent behavior. For the SQG, due to the influence of $\alpha$ the value $ |\Lambda^{-\frac{\alpha}{2}} g| $ lies in the range $ 4.935 \times 10^{-5} $ to $1.5942 \times 10^{-4}$.  When computing spectra for the SQG we use different viscosity values depending on the specific value of $\alpha$ to reach the dissipation range.  A more detailed explanation of how we select suitable viscosity ranges is given below. The Grashof numbers remain large (ranging from $ 6.377 \times 10^5 $ to $ 1.235 \times 10^{10} $). Similarly, in the case of the gSQG, we vary the viscosity depending on $ \alpha $ and $ \beta $. 

For all runs, we start with the same randomly chosen initial condition $\theta(-20{,}000) = \theta_0$,  and expect that by $t = 0$, when the averaging begins, the transient phase has already passed. The resolution is enhanced as time increases, as shown in the time series plots in Figure~\ref{Fig:Time_series_plot} (left).  We include a physical space plot in Figure~\ref{Fig:Time_series_plot} (right), taken at the time when the $\eta$ value reaches its maximum over the period when $N = 16384$.

\begin{figure}[H]
\centerline{
\includegraphics[width=8.5cm, height=6.5cm]{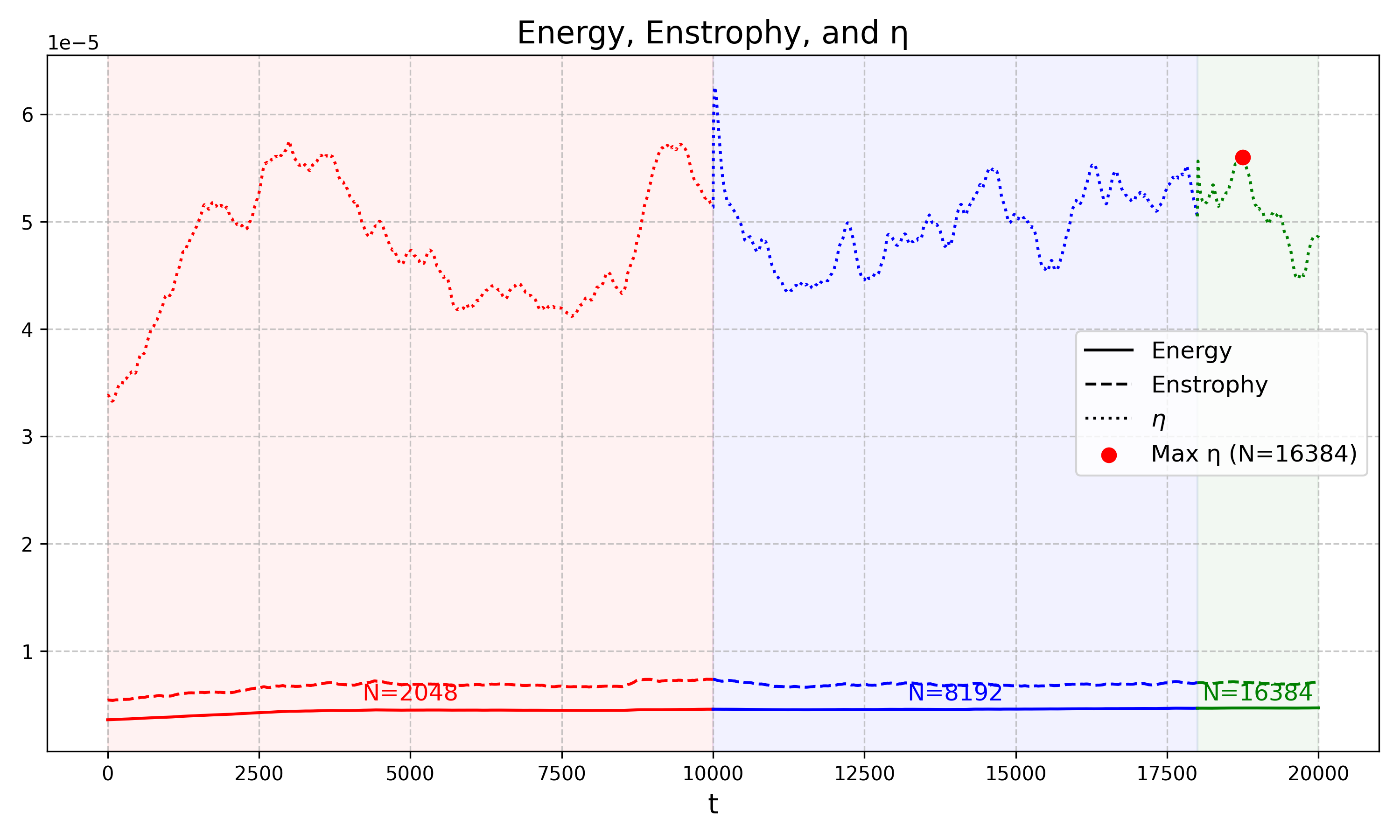}\
\includegraphics[width=8.5cm, height=6.5cm]{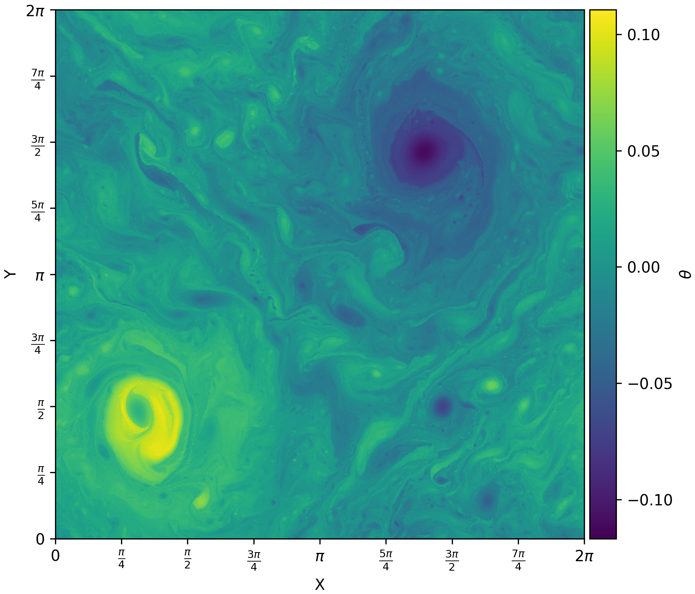}}
\caption{Left:
Time series plot for energy, enstrophy and $\eta$.  Right: Physical space plot at time 
of maximal $\eta$.}
\label{Fig:Time_series_plot}
\end{figure}

\subsection{Computed energy spectra} \label{computedspectrum}
Each computed energy spectra for the NSE in Figure~\ref{Fig:NSEspec} (left)  exhibit a clean $\kappa^{-3}$ scaling within some inertial range that extends at least through $\kappa_\eta$.   Starting at roughly $10\times \kappa_\eta$  there is a rapid fall-off characteristic of a dissipation range.  As expected, smaller viscosities yield wider inertial ranges.  The nearly horizontal plots of the compensated spectra in Figure ~\ref{Fig:NSEspec} (right) show the faithfulness with the power law more clearly.

\setlength{\intextsep}{1pt}
\begin{figure}[H]
\centerline{\includegraphics[width=8.5cm, height=6.5cm]{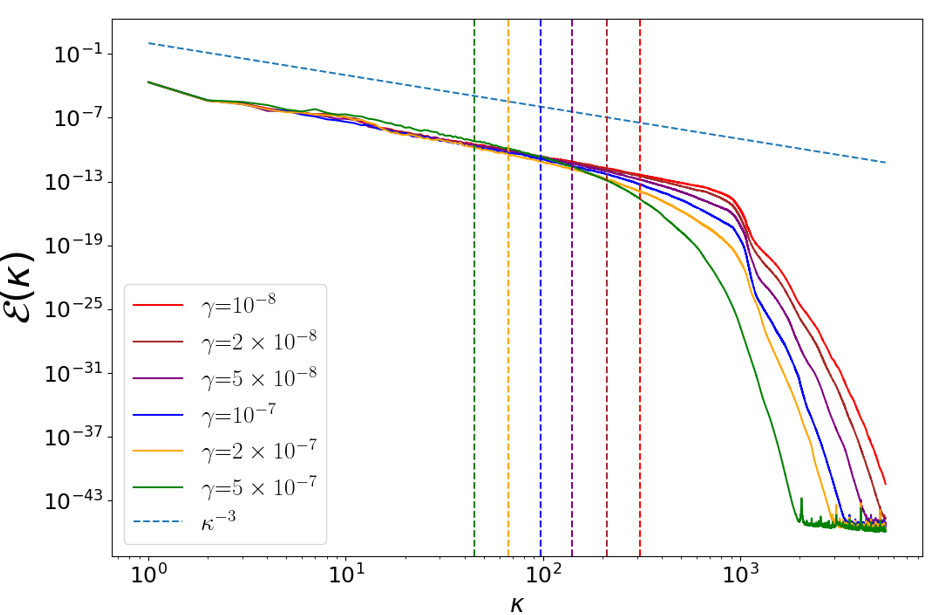}\ 
\includegraphics[width=8.5cm, height=6.5cm]{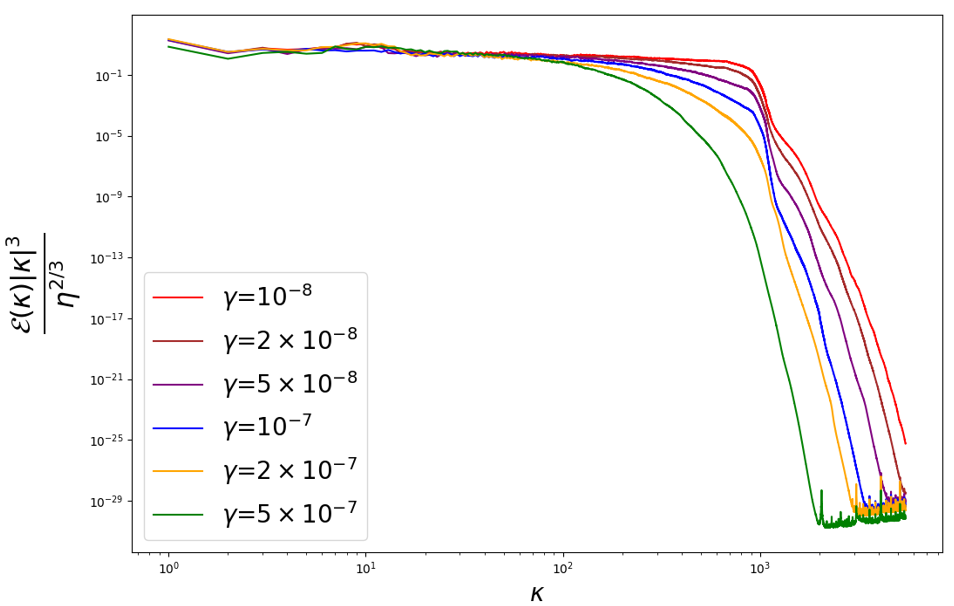}}
\caption{Left: energy spectra for the NSE.  Right: compensated spectra for NSE.  The values of $\ka_\eta$ are indicated by vertical lines, $N=16384$.}
\label{Fig:NSEspec}
\end{figure}

For the gSQG  we vary $\ga$ and $\alpha$ together to keep $\kappa_\eta$ within the interval [50,1000].  If $\kappa_\eta < 50 $, there would not be enough room for a significant inertial range. On the other hand, if $ \kappa_\eta > 1000 $, the flow would not be adequately resolved in that the dissipation range is beyond $10\times \kappa_\eta$.
In general, when the external force and all other conditions are fixed, a smaller $ \alpha $ tends to produce a larger $ \kappa_\eta $. To keep $ \kappa_\eta $ within a practical range, we use larger viscosity values for smaller $ \alpha $.  Based on this criterion, we choose the following values of $ \gamma $ for the SQG cases, as summarized in the table below.

\setlength{\intextsep}{1pt}
\begin{figure}[H]
\centerline{
\includegraphics[width=8cm, height=6.5cm]{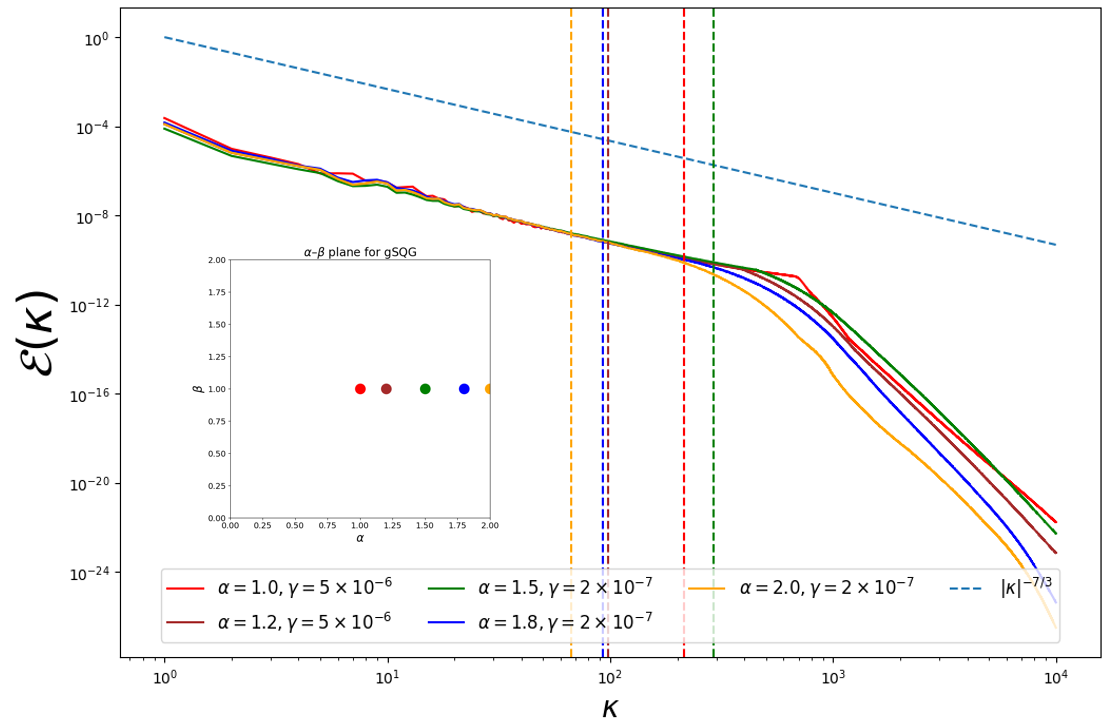} \
\includegraphics[width=8cm, height=6.5cm]{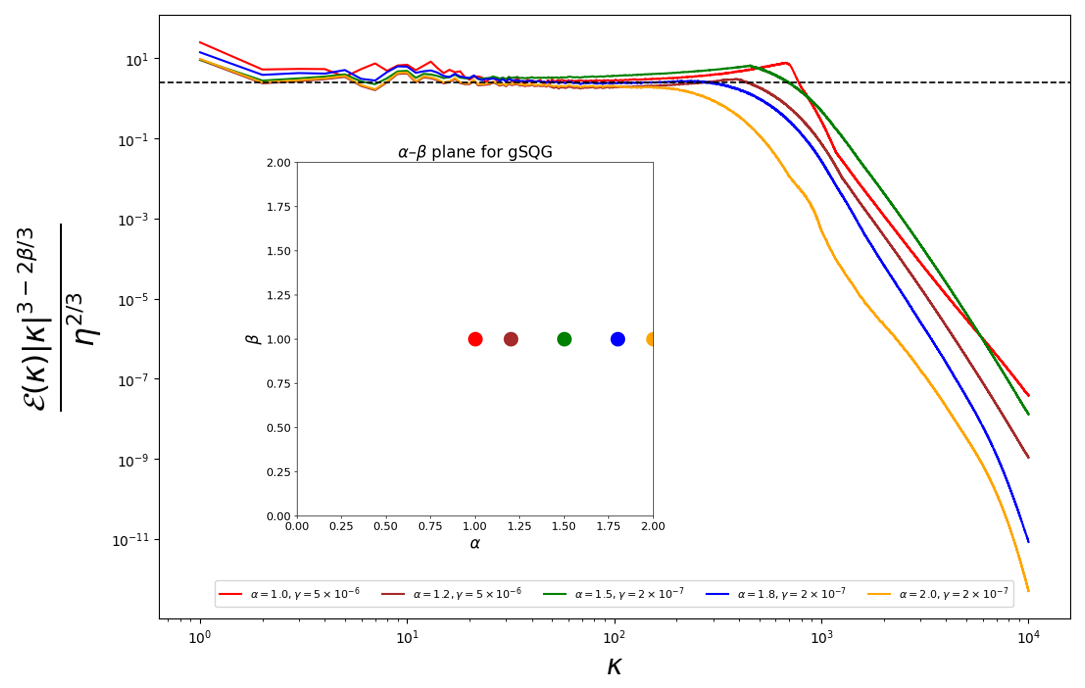}
}
\caption{Left: energy spectra for the SQG. Right: compensated spectra for SQG. The values of $\kappa_\eta$ are indicated by vertical lines, $N = 32768$.}
\label{Fig:SQGspec}
\end{figure}

\vspace{0.1in}

\setlength{\intextsep}{1pt}
\begin{figure}[H]
\centerline{
\includegraphics[width=5cm, height=4cm]{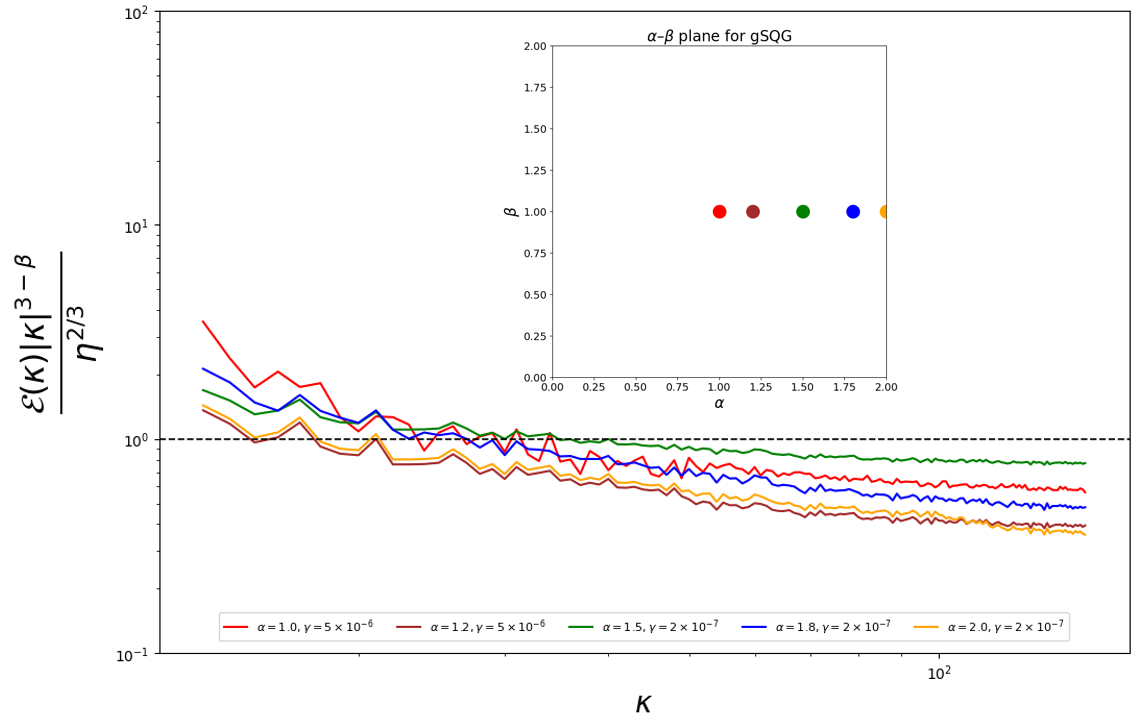} \
\includegraphics[width=5cm, height=4cm]{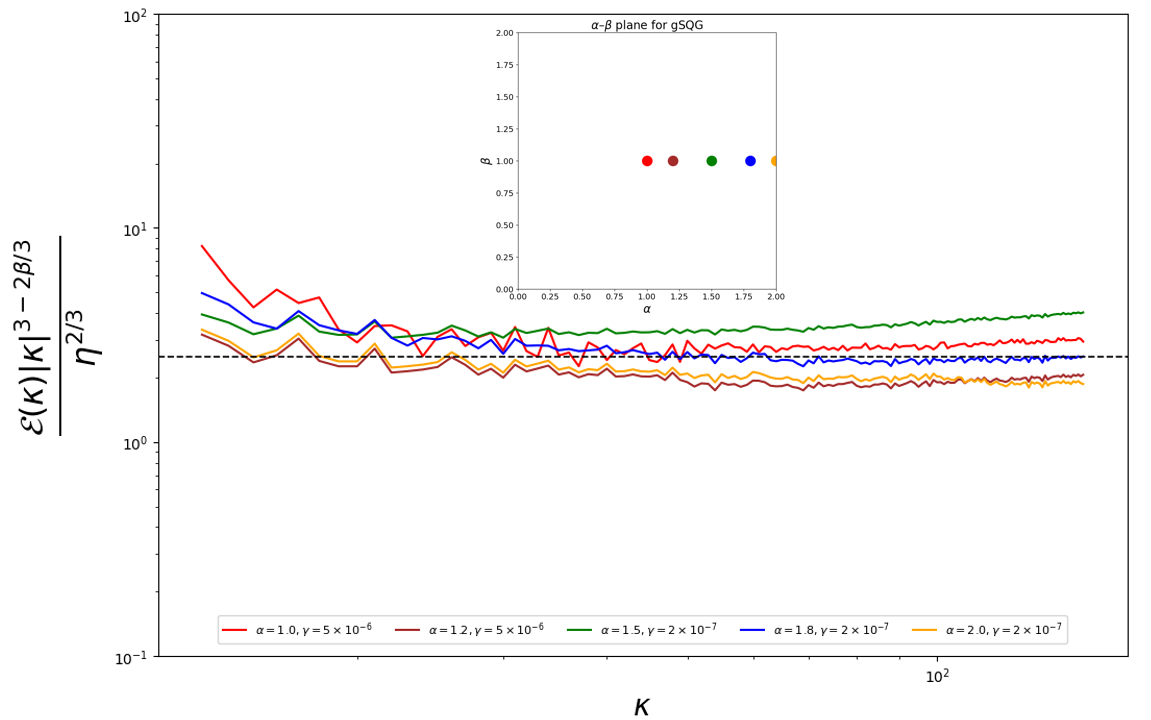} \
\includegraphics[width=5cm, height=4cm]{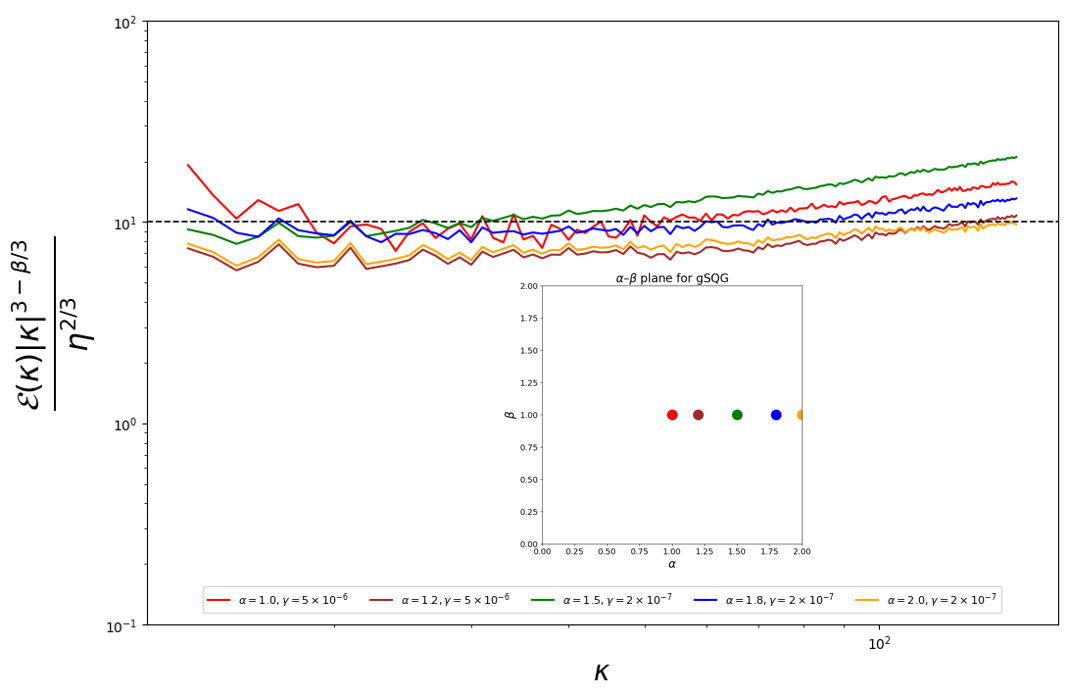}
}
\caption{Compensated spectrum plots for $\ka \in [12,150]$ with different compensated slopes Left: slope $3-\beta$ . Middle: slope $3-\frac{2\beta}{3}$ . Right:  $3-\frac{\beta}{3}$ as in \cite{Pierrehumbert1994}.}
\label{Fig:SQGspecs23}
\end{figure}

\vspace{0.1in}

\begin{table}[h!]
\centering
\begin{tabular}{|c|c|c|c|c|c|}
\hline
SQG & $\alpha = 1.0$ & $\alpha = 1.2$ & $\alpha = 1.5$ & $\alpha = 1.8$ & $\alpha = 2.0$ \\
\hline
$\gamma$ & $5 \times 10^{-6}$ & $5 \times 10^{-6}$ & $2 \times 10^{-7}$ & $2 \times 10^{-7}$ & $2 \times 10^{-7}$ \\
\hline
$\kappa_{\eta}$ & 214 & 98 & 292 & 93 & 67 \\
\hline
$\kappa_{\sigma}$ & 27 & 23 & 59 & 25 & 20 \\
\hline
\end{tabular}
\caption{Chosen viscosity values and corresponding $ \kappa_\eta $ for different $ \alpha $ values}
\end{table}

Figure~\ref{Fig:SQGspec}  shows that the spectra for the  subcritical SQG for a sampling of $\alpha$ values in $(1,2]$ nearly obey the power law in \eqref{main:ES} for $\beta=1$, at least up to $\kappa_\eta$.  The compensated spectra are not as close to horizontal as in the case of the NSE.  Note that there is a sharp corner formed at the start of the dissipation fall-off in the case of the critical SQG.  This presages a deviation from the expected power law for the gSQG when we take $\alpha<1$.  Considering the cases of $ \alpha = 1.0, 1.2 $ and $ \alpha = 1.5, 1.8, 2.0 $ separately, we note that for fixed viscosity, $ \kappa_\eta $ increases as $\alpha$ decreases.

We compare the compensated spectra in \eqref{main:ES}, \eqref{main:ES2} and \eqref{main:ES3} over a plausible inertial range in Figure \ref{Fig:SQGspecs23} and find that 
$\mathcal{E}(\kappa) \sim \eta^{2/3}\ka^{2\beta/{3}-3}$
fits the computed data best. This spectrum will be assumed as a condition in the rigrorous estimate in Proposition \ref{main:prop2}.

For the gSQG we sample along the diagonal line in the $ \alpha, \beta $ plane from $ \alpha = 1.8 $, $ \beta = 0.2 $ to $ \alpha = 0.1 $, $ \beta = 1.9$.  In order to ensure adequate resolution, $ \gamma $ is varied along with $\alpha$ and $\beta$ with the resulting values of $\kappa_\eta$ shown in Table \ref{tab:gsqg_parameters}.  Figure~\ref{Fig:gSQGspectra} (left) shows some adherence to the heuristic power law \eqref{main:ES}.   For the samples in the southeast corner ($\alpha\ge 1$) the spectra are similar to those for the SQG, though we note a difference in $\kappa_\eta$ values.  Compared to the SQG, $\kappa_\eta$ is larger on this diagonal for 
$\alpha=1.2$, and smaller for $\alpha=1.5$. Since the expected slopes vary with $\beta$,  plots of the compensated spectra are shown in Figure~\ref{Fig:gSQGspectra} (right).

\vspace{0.1in}

\begin{table}[h!]
\centering
\small
\begin{tabular}{|c|c|c|c|c|c|c|}
\hline
gSQG & \( \alpha=0.2,\beta=1.8 \) & \( \alpha=0.5,\beta=1.5 \) & \( \alpha=0.8,\beta=1.2 \) & \( \alpha=1.2,\beta=0.8 \) & \( \alpha=1.5,\beta=0.5 \) & \( \alpha=1.8,\beta=0.2 \) \\
\hline
\( \gamma \) & \( 5 \times 10^{-4} \) & \( 1 \times 10^{-4} \) & \( 5 \times 10^{-6} \) & \( 5 \times 10^{-6} \) & \( 2 \times 10^{-7} \) & \( 2 \times 10^{-7} \) \\
\hline
\( \kappa_{\eta} \) & 683 & 361 & 665 & 145 & 206 & 114 \\
\hline
\( \kappa_{\sigma} \) & 3 & 6 & 27 & 18 & 48 & 34 \\
\hline
\end{tabular}
\caption{Viscosity \( \gamma \) and dissipation scale \( \kappa_{\eta} \) for various \( (\alpha, \beta) \) in gSQG}
\label{tab:gsqg_parameters}
\end{table}

The discrepancy in the spectrum for the critical SQG is found to persist for the three samples in the northwest corner ($\alpha <1$).  For the two on and beyond the critical line $\beta=\alpha+1$,  where the gSQG changes from quasilinear to fully nonlinear,  the compensated spectra actually increase before falling off at a slow rate.  A comparison of compensated plots over the inertial range in Figure \ref{Fig:gSQGspecs23} again shows a reasonable fit for the power law $3-2\beta/3$.

\setlength{\intextsep}{1pt}
\begin{figure}[H]
\centerline{\includegraphics[width=8.5cm, height=6.5cm]{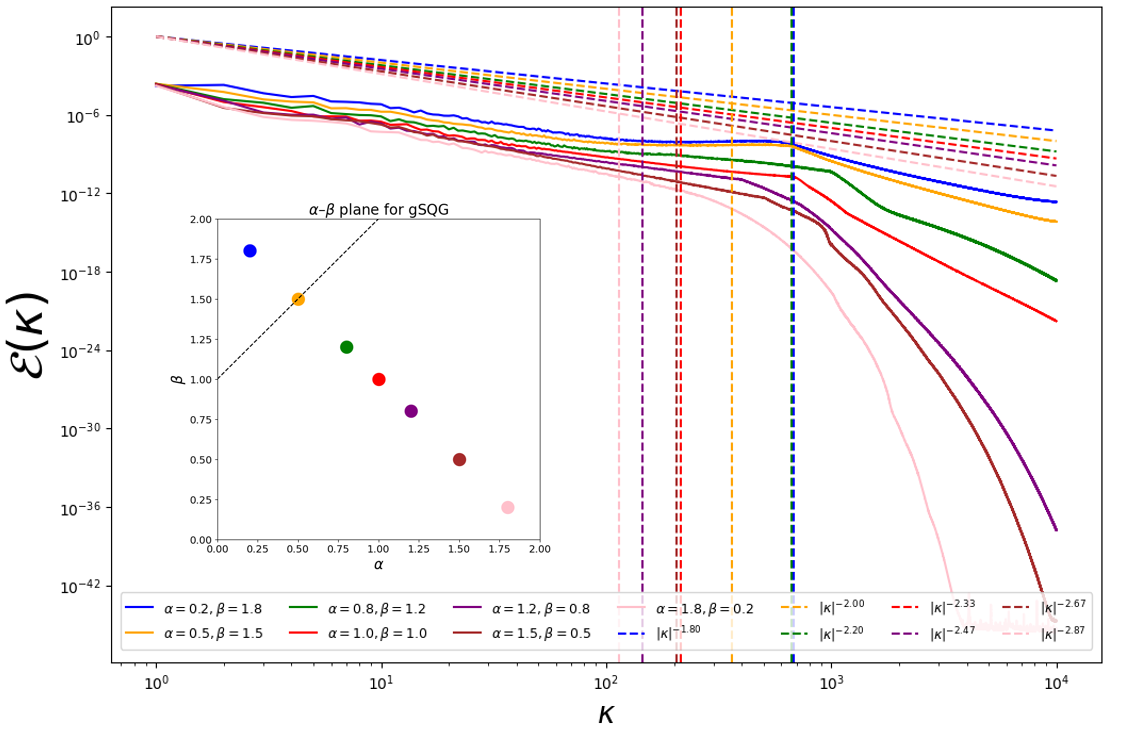}\ 
\includegraphics[width=8.5cm, height=6.5cm]{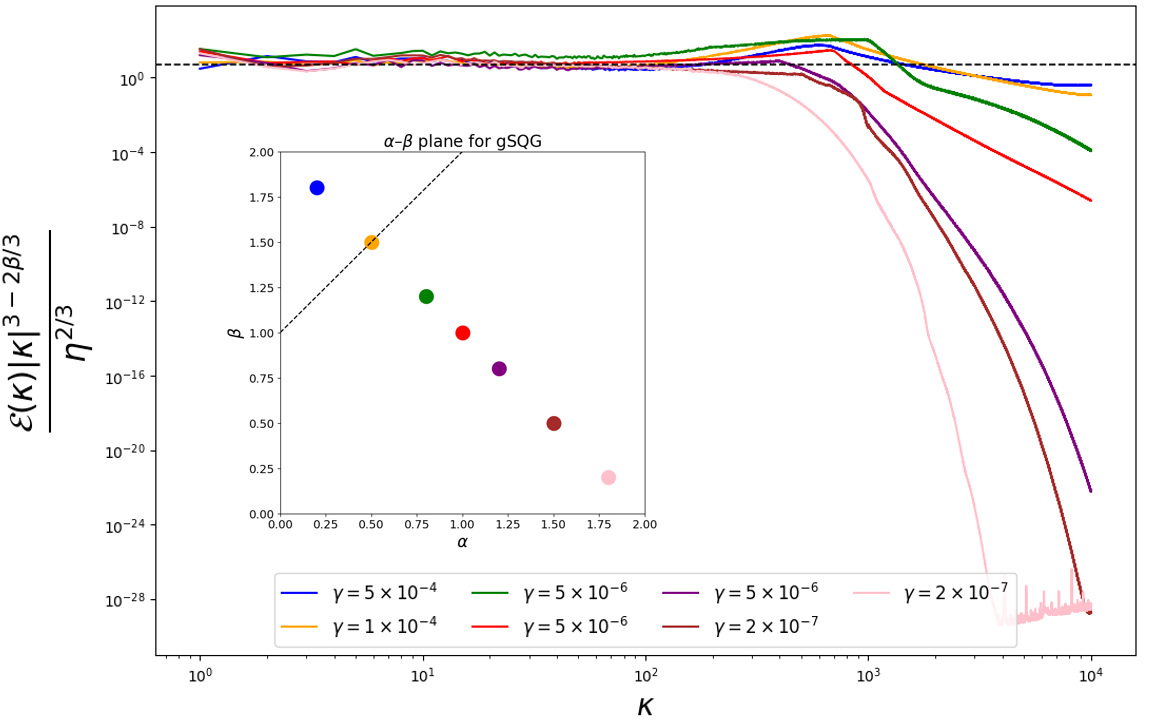}}
\caption{Left: energy spectra for the gSQG.  Right: compensated energy spectra. The values of $\ka_\eta$ are indicated by vertical lines, N=32768.}
\label{Fig:gSQGspectra}
\end{figure}

\setlength{\intextsep}{1pt}
\begin{figure}[H]
\centerline{
\includegraphics[width=5cm, height=4cm]{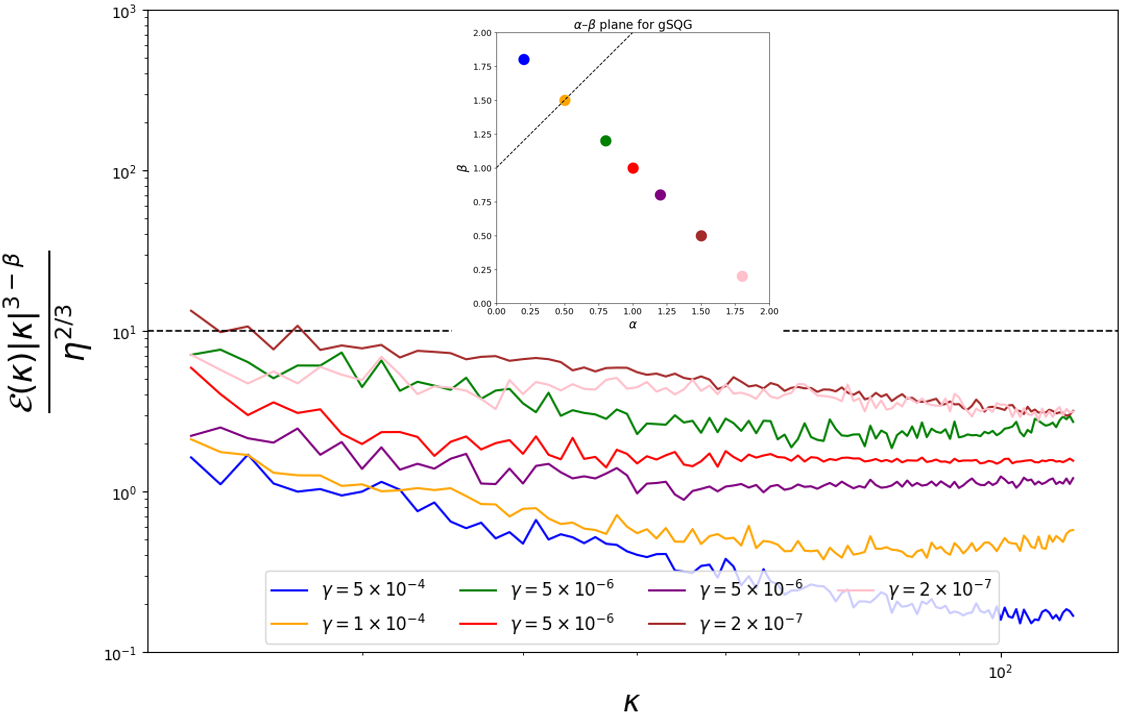} \
\includegraphics[width=5cm, height=4cm]{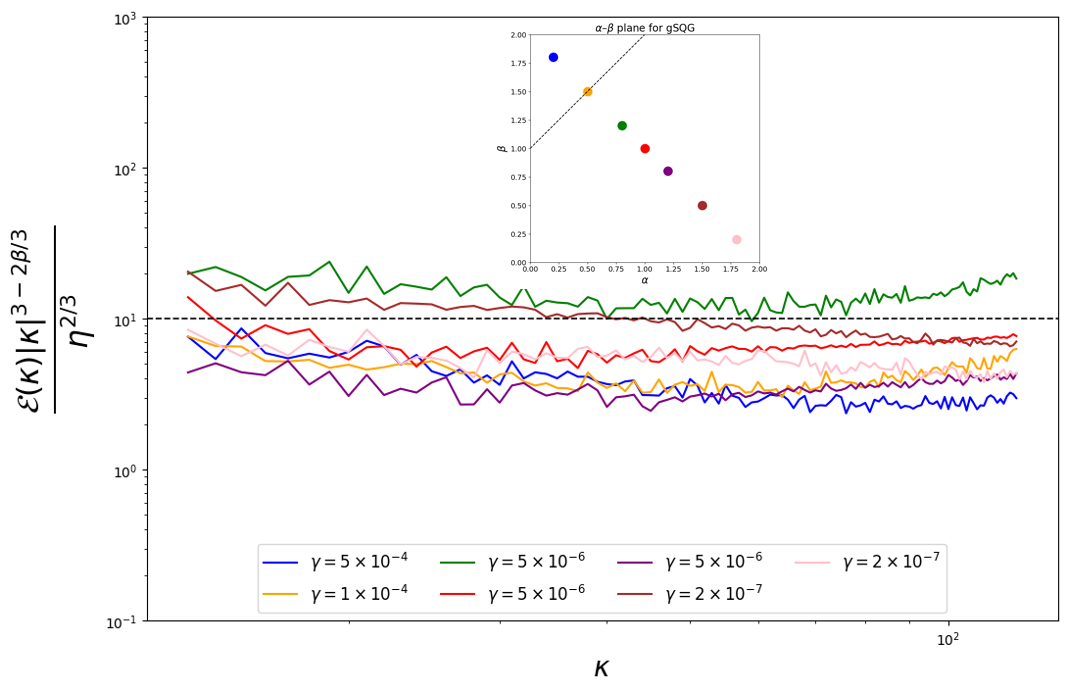} \
\includegraphics[width=5cm, height=4cm]{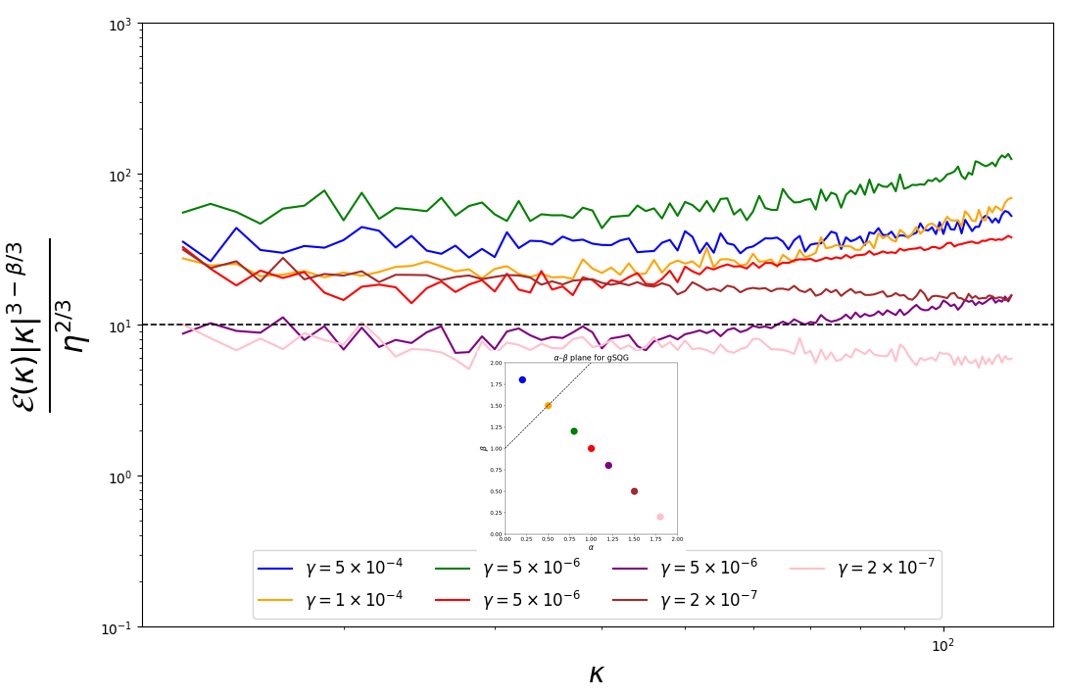}
}
\caption{Compensated spectrum plots for $\ka \in [12,150]$ with different compensated slopes Left: slope $3-\beta$ . Middle: slope $3-{2\beta}/{3}$ . Right:  $3-{\beta}/{3}$}
\label{Fig:gSQGspecs23}
\end{figure}

\section{The Enstrophy Cascade}
\subsection{Analytical Support}
Let $p_{\kappa}= P_{\kappa}\te \text{  and  } q_{\kappa}= Q_{\kappa}\te$.
{Taking the $H$-inner product of  \eqref{maineq} with} $p_{\kappa}$ 
and applying the first relation in \eqref{bterm_prod},
we obtain
\begin{equation} \label{prod_E_k}
\begin{split} 
        \frac{1}{2} \frac{d}{dt} |p_{\kappa}|^2 + \ga |\la^{\frac{\alpha}{2}} p_{\kappa}|^2 & = (P_{\kappa}u\cdotp\nabla p_{\kappa},q_{\kappa}) - (Q_{\kappa}u\cdotp\nabla q_{\kappa},p_{\kappa}) + (g,p_{\kappa}) \\
        & = - L^2\mathfrak{E}_{\kappa} + (g,p_{\kappa}) .
\end{split}    
\end{equation}
{While taking the $H$-inner product of \eqref{maineq} instead with} $q_{\kappa}$, we find
\begin{equation} \label{prod_e_k}
\begin{split}
        \frac{1}{2} \frac{d}{dt} |q_{\kappa}|^2 + \ga |\la^{\frac{\alpha}{2}} q_{\kappa}|^2 & = (Q_{\kappa}u\cdotp\nabla q_{\kappa},p_{\kappa}) - (P_{\kappa}u\cdotp\nabla p_{\kappa},q_{\kappa}) + (g,q_{\kappa}) \\
        & = L^2\mathfrak{E}_{\kappa} + (g,q_{\kappa}) ,
\end{split}    
\end{equation}
{In both of the above energy balances,  $\mathfrak{E}_{\kappa}$ is the quantity defined by
    \begin{align}\label{def:enstrophy:flux}
        \mathfrak{E}_{\kappa}\overset{\text{def}}{=}(Q_{\kappa}u \cdotp\nabla q_{\kappa},p_{\kappa}) - (P_{\kappa}u \cdotp\nabla p_{\kappa},q_{\kappa}),
    \end{align}
and it is interpreted to be the} \textit{net rate of enstrophy transfer} (or \textit{net enstrophy flux}) at the wavenumber $\kappa$, from the low modes to the high modes.

Similarly, {upon taking $H$-inner product of \eqref{maineq} with $-P_{\kappa}\psi$ (respectively $-Q_{\kappa}\psi$)}, we obtain
\begin{equation*}
        \frac{1}{2} \frac{d}{dt} |\la^{\frac{\beta-2}{2}} P_{\kappa}\te|^2 + \ga|\la^{\frac{\alpha+\beta-2}{2}} P_{\kappa}\te|^2 =  -L^2\mathfrak{e}_{\kappa} + (g,P_{\kappa}\la^{\beta-2}\te)
\end{equation*}
\begin{equation*}
        \frac{1}{2} \frac{d}{dt} |\la^{\frac{\beta-2}{2}} Q_{\kappa}\te|^2 + \ga|\la^{\frac{\alpha+\beta-2}{2}} Q_{\kappa}\te|^2 = L^2\mathfrak{e}_{\kappa} + (g,Q_{\kappa}\la^{\beta-2}\te)\;,
\end{equation*}
{where $\mathfrak{e}_{\kappa}$ is defined by
    \begin{align}\label{def:energy:flux}
        \mathfrak{e}_{\kappa}\overset{\text{def}}{=}(Q_\kap u\cdotp\nabla q_\kap,\Lam^{\be-2}p_\kap)-(P_\kap u\cdotp\nabla p_\kap,\Lam^{\be-2}q_\kap).
    \end{align}
Then $\mathfrak{e}_{\kappa}$ is interpreted as} the {\it net rate of energy transfer} (or {\it net energy flux}) at the wavenumber $\kappa$ {from low modes to high modes}. For additional details in deriving \eqref{def:energy:flux}, the reader is referred to Appendix \ref{app:energy}.

The following {result} states that for wavenumbers smaller than the injection range of the force, both enstrophy and energy are transferred to smaller wavenumbers. {The proof is a straightforward adaptation of the analogous result for the NSE which can be found in \cite{FJMR2002}.}

\begin{proposition} \label{def:mathfraks}
    Suppose that $\kappa \leq {\kapb}$, then the time average of net fluxes \eqref{prod_E_k} and \eqref{prod_e_k} satisfy
    \begin{equation} \label{mathfrakE_k}
        \langle\E_{\kappa}(\te)\rangle = -\frac{\ga}{L^2} \langle|\la^{\frac{\alpha}{2}} P_{\kappa}\te|^2\rangle
    \end{equation}
    and
    \begin{equation} \label{mathfrake_k}
        \langle\mathfrak{e}_{\kappa}(\te)\rangle = -\frac{\ga}{L^2} \langle|\la^{\frac{\alpha+\beta-2}{2}} P_{\kappa}\te|^2\rangle .
    \end{equation}
\end{proposition}
\begin{proof}
We take the time average of \eqref{prod_E_k} to obtain
    \begin{equation*}
        \frac{1}{2t} (|p_{\kappa}(t)|^2 - |p_{\kappa}(0)|^2) + \frac{\ga}{t} \int_0^t |\la^{\frac{\alpha}{2}} p_{\kappa}(\tau)|^2 d\tau = -\frac{L^2}{t} \int_0^t \mathfrak{E}_{\ka}(\te(\tau))d\tau .
    \end{equation*}
    We take $t \to \infty$ and by the boundedness of ${|p_\ka|}$, find
    \begin{equation*}
        \lim_{t\to \infty} \frac{\ga}{t} \int_0^t |\la^{\frac{\alpha}{2}} p_{\kappa}(\tau)|^2 d\tau = \lim_{t\to \infty} -\frac{L^2}{t} \int_0^t \mathfrak{E}_{\ka}(\te(\tau))d\tau
    \end{equation*}
    from which \eqref{mathfrakE_k} follows.
    A similar procedure applied to  \eqref{prod_e_k} yields \eqref{mathfrake_k}.
\end{proof}
At wavenumbers larger than the injection range, the {direction of transfer switches} to larger wavenumbers.
\begin{proposition}\label{directflux}
    If $\kappa \geq \bkap$, then the time average of net fluxes satisfy
    \begin{equation*}
        \langle\E_{\kappa}(\te)\rangle = \frac{\ga}{L^2} \langle|\la^{\frac{\alpha}{2}} Q_{\kappa}\te|^2\rangle
    \end{equation*}
    and
    \begin{equation*}
        \langle\mathfrak{e}_{\kappa}(\te)\rangle = \frac{\ga}{L^2} \langle|\la^{\frac{\alpha+\beta-2}{2}} Q_{\kappa}\te|^2\rangle \;.
    \end{equation*}
\end{proposition}
\begin{proof}
    {One may argue in a similar fashion to the proof of Proposition~\ref{def:mathfraks}, except that one instead considers the high-mode energy (respectively enstrophy) balance; we omit the details.}
\end{proof}


By a \textit{direct cascade} of enstrophy, we mean that {the net enstrophy flux}, $\langle \E_\kappa \rangle$, is positive {(``direct")} and roughly constant {(``cascade")}  over some range in $\kappa$.  That it is positive {above the injection scale determined by $\bkap$} is guaranteed  by Proposition \ref{directflux}.  The expectation is that the enstrophy injected at $\bar\kappa$ should match the total that is dissipated, so that  $\langle \E_\kappa \rangle\approx \eta$ over this range.  The extent to which this approximation holds then quantifies how pronounced the direct enstrophy cascade is.  In contrast, an {\it inverse cascade} refers to a {roughly constant} transfer {towards} smaller wavenumbers {from higher wavenumbers. Proposition \ref{def:mathfraks} guarantees the correct direction of \textit{energy} transfer below the injection scale determined by $\kapb$, namely, that the net energy flux, $\lb \mathfrak{e}_\kap\rb$, is negative for $\kap\leq\kapb$.}
{Although Proposition \ref{def:mathfraks} and Proposition \ref{directflux} indicate that both} energy and enstrophy {are directed in the same way through scales}, a more detailed analysis in the case of the NSE in \cite{FJMR2002} shows that the direct cascade of enstrophy is {actually} more pronounced than that of energy, {while} the inverse cascade {of energy is more pronounced than that of enstrophy. In what follows, we will focus} on the \textit{direct cascade of enstrophy}.   

Let
\begin{equation} \label{def:ksigma}
    \kappa_{\sigma} = \left(\frac{\langle|\la^{\frac{\alpha}{2}}\te|^2\rangle}{\langle|\te|^2\rangle}\right)^ {1/\alpha} .
\end{equation}
{We claim that this} wavenumber can serve as an indicator for a strongly pronounced direct enstrophy cascade.  

\begin{proposition}
    \label{main_prop1}
For ${\kap>\bkap}$ we have
\begin{equation} \label{main:ROETsidethem}
    1-\left(\frac{\kappa}{\kappa_{\sigma}}\right)^{\alpha}
    \leq \frac{\langle \E_{\kappa}\rangle}{\eta} \leq 1 .
\end{equation}
Hence, given $\delta \in (0,1]$, and $\kappa \in (\bkap, \delta^{1/\alpha}\kappa_{\sigma}]$, 
\begin{equation*}
   0 \le 1 - \frac{\langle \E_{\kappa}\rangle}{\eta} \leq \delta .
\end{equation*}
\end{proposition} 
{In particular, this implies that
$ \langle \E_{\kappa}\rangle \approx \eta $ over the range ${\bkap < \kappa \ll \kappa_{\sigma}}$, whenever ${\kappa_{\sigma} \gg \bkap}$}.  The proof of Proposition 
\ref{main_prop1} is another straightforward adaptation of one for the NSE case in \cite{FJMR2002}.  Since it is brief, we include it for completeness.

\begin{proof}
Note that by the {Bernstein inequalities (Lemma \ref{T:Bernstein})}
\begin{equation*}
    \ga\langle|\la^{\frac{\alpha}{2}} P_{\kappa}\te|^2\rangle \leq \ga \kappa^{\alpha} \langle|P_{\kappa}\te|^2\rangle \leq  \ga \kappa^{\alpha}\langle|\te|^2\rangle \leq \ga \left(\frac{\kappa}{\kappa_{\sigma}}\right)^{\alpha}\langle|\la^{\frac{\alpha}{2}}\te|^2\rangle
\end{equation*}
and consequently, for $\kappa > \bkap$ we obtain that
\begin{equation*}
    L^2\langle \E_{\kappa}\rangle = \ga \langle|\la^{\frac{\alpha}{2}} Q_{\kappa}\te|^2\rangle = \ga \langle|\la^{\frac{\alpha}{2}}\te|^2\rangle - \ga \langle|\la^{\frac{\alpha}{2}} P_{\kappa}\te|^2\rangle \geq \ga \langle|\la^{\frac{\alpha}{2}}\te|^2\rangle \left(1-\left(\frac{\kappa}{\kappa_{\sigma}}\right)^{\alpha}\right)\;.
\end{equation*}
\end{proof}

\subsection{Computed enstrophy flux}

To gauge the strength of the enstrophy cascade we plot the compensated flux $\frac{\langle \E_{\kappa} \rangle}{\eta}$.
The result in the NSE case for varying values of $\gamma$ is shown in Figure~\ref{Fig:combined_ROET} (left).  As expected, the enstrophy cascade is more pronounced as the viscosity decreases, which is consistent with the increase in $\kappa_\sigma$.  Considering the $\kappa_\eta$ values in Figure \ref{Fig:NSEspec} (left), we see that ${\langle \E_{\kappa} \rangle}\approx {\eta}$ holds
for $\bar\kappa \le \kappa \le \kappa_\eta$.

For the SQG case, we fix \( \gamma = 10^{-7} \) and vary \( \alpha \) in Figure~\ref{Fig:combined_ROET} (right). The reason we choose a relatively small \( \gamma \) is that, if we use the viscosities as in the spectrum cases, the flow is fully resolved, but with small \( \kappa_\sigma \) values.   At smaller  \( \gamma \) values we find \( \kappa_\sigma \) to be large.  Since the spectra have not reached the dissipation fall-off at smaller  \( \gamma \) values, we expect that the values of  \( \kappa_\sigma \) would only be larger if the equation were fully resolved.  

Here we see the effect of the more singular constitutive law of the SQG.  The values of $\kappa_\sigma$ increase from nearly 40 to over 100 as $\alpha$ decreases, compared to a little over 30 for the NSE at the same viscosity.  This is consistent with the plots of the compensated flux which show a widening of the cascade range as the critical SQG is approached.  

\setlength{\intextsep}{1pt}
\begin{figure}[H]
\centerline{
\includegraphics[width=8.5cm, height=6.5cm]{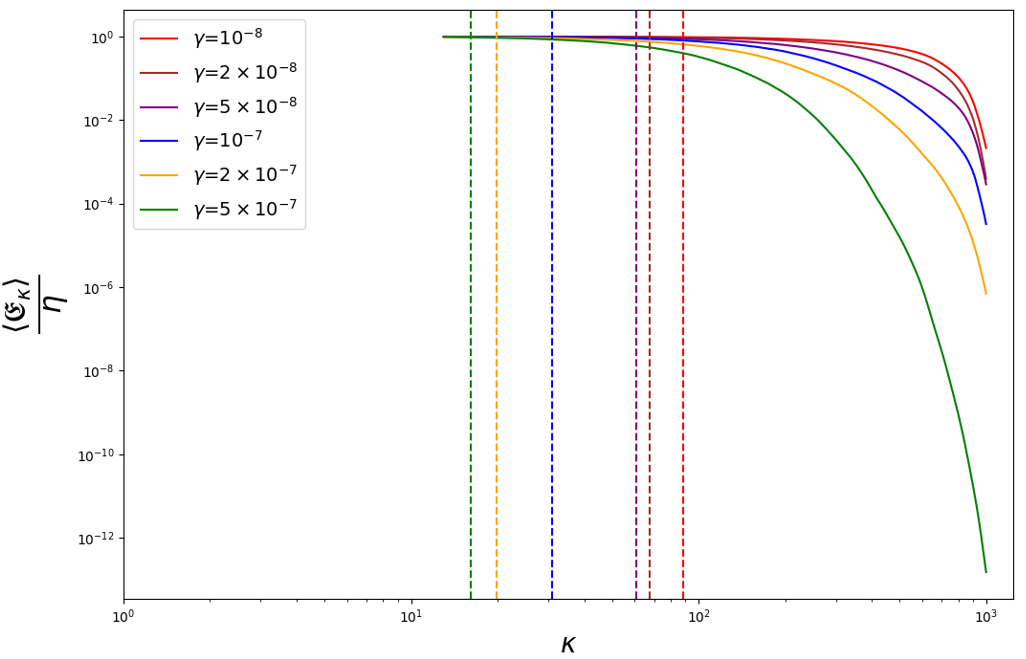} \
\includegraphics[width=8.5cm, height=6.5cm]{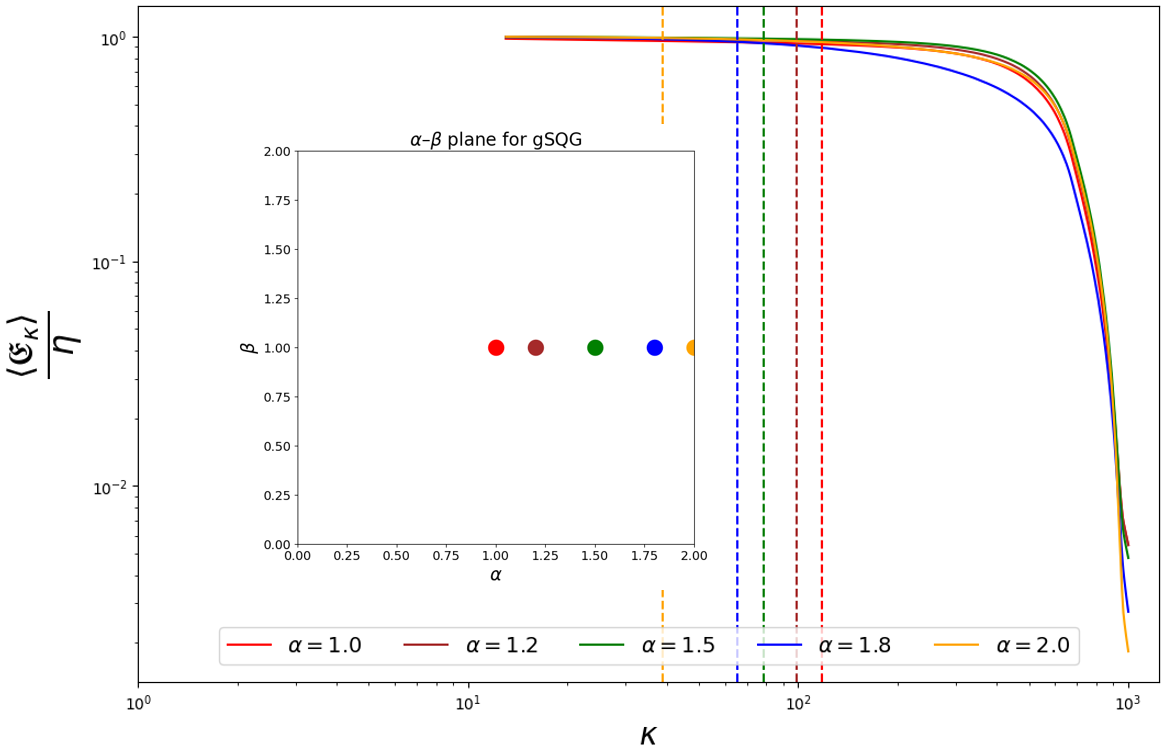}
}
\caption{
Compensated enstrophy flux.  Left: NSE, $N=16384$.  Right: SQG, $\gamma=10^{-7}$, $N=32768$.
 Dashed vertical lines indicate $ \kappa_{\sigma} $.}
\label{Fig:combined_ROET}
\end{figure}

We keep $\gamma = 10^{-7}$ and plot in Figure \ref{Fig:gSQG_DIA_ROET} the compensated flux for the same samples on the diagonal in the $\alpha,\beta$-plane as for the spectrum of the gSQG.
In the southeast quadrant as $ \alpha $ decreases and $ \beta $ increases, the value of $ \kappa_{\sigma} $ increases, reaching a peak in the case of the critical SQG $ \alpha = 1 $, $ \beta = 1 $.  Correspondingly, the enstrophy cascade strengthens as $\alpha$ decreases.  However, in the northwest quadrant (where $ \alpha \in (0, 1) $ and $ \beta \in (1, 2) $ ), this trend reverses: $ \kappa_{\sigma} $ decreases as $ \alpha $ decreases and the cascade weakens.

\setlength{\intextsep}{1pt}
\begin{figure}[H]
\centerline{
\includegraphics[width=8.5cm, height=6.5cm]{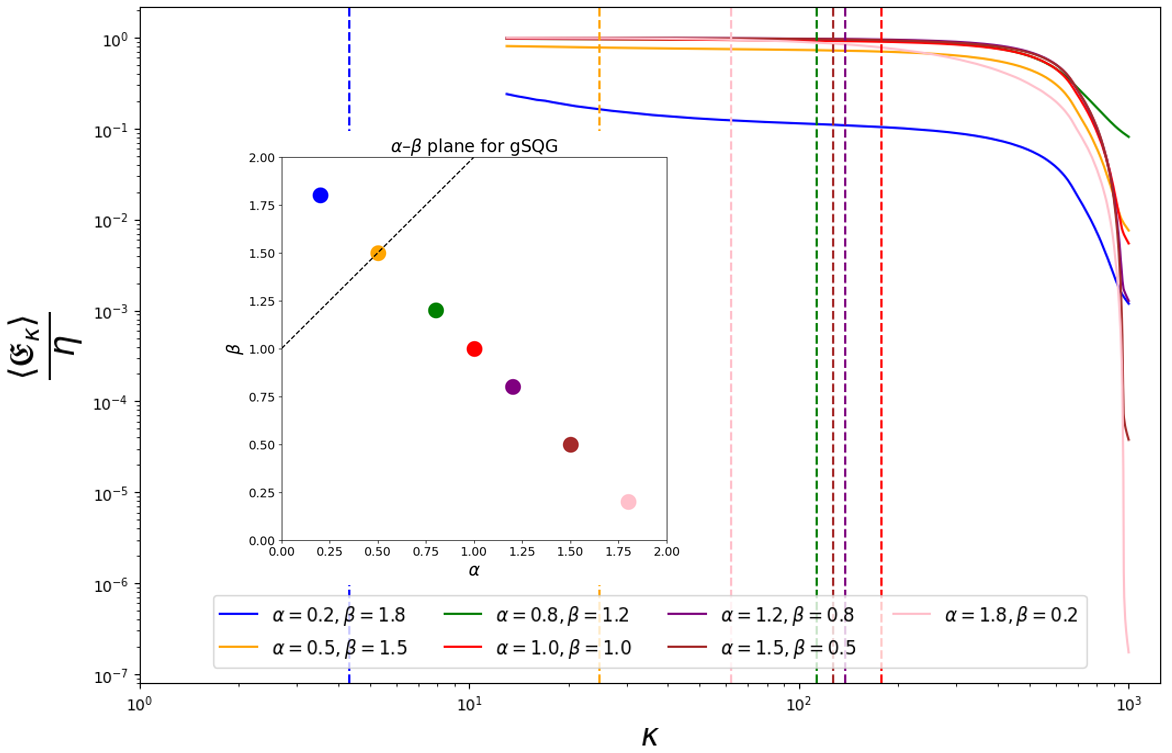} \ 
\includegraphics[width=8.5cm, height=6.5cm]{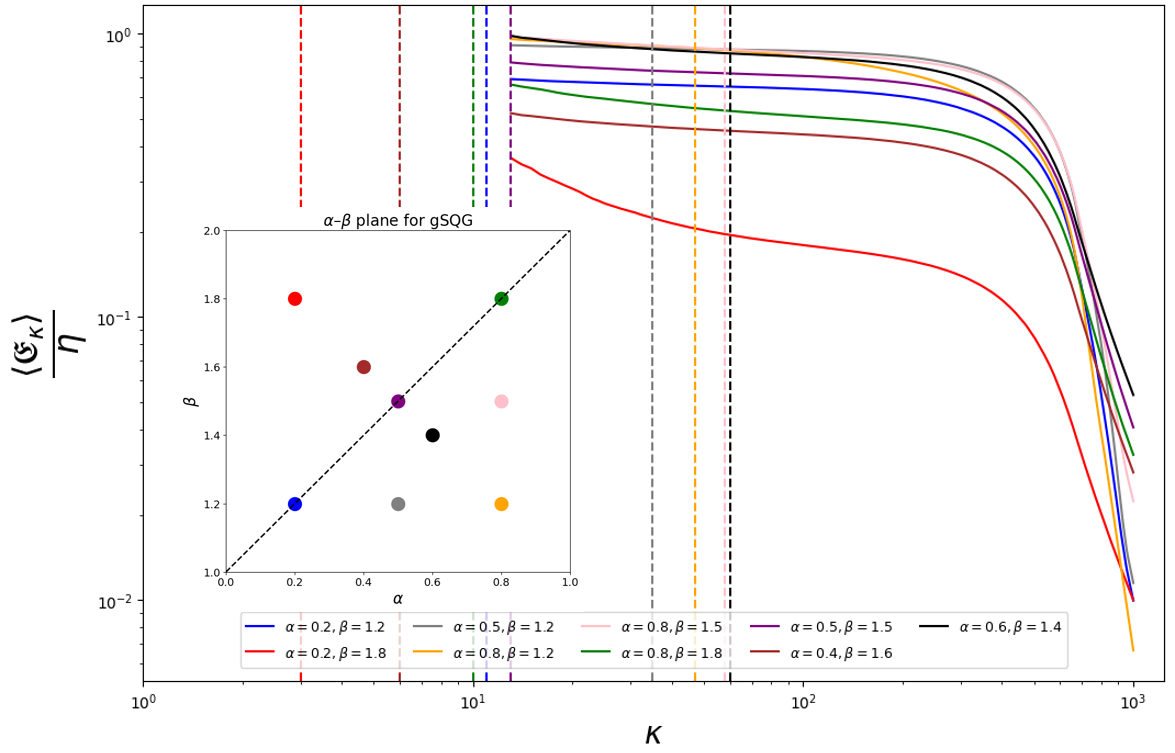}}
\caption{Compensated enstrophy flux for the gSQG. Dashed vertical lines indicate $ \kappa_{\sigma} $, $N=32768$.   $\gamma=10^{-7}$. }
\label{Fig:gSQG_DIA_ROET}
\end{figure}

In Figure \ref{fig:ROET_comparison} (left) we zoom in to observe the effect of $\alpha$ on the lower bound estimate in \eqref{main:ROETsidethem}.  This bounding curve becomes less concave as $\alpha$ decreases, until it is linear at $\alpha = 1$, and convex in the NW quadrant.

\setlength{\intextsep}{1pt}
\begin{figure}[H]
\centerline{
\includegraphics[width=8.5cm, height=6.5cm]{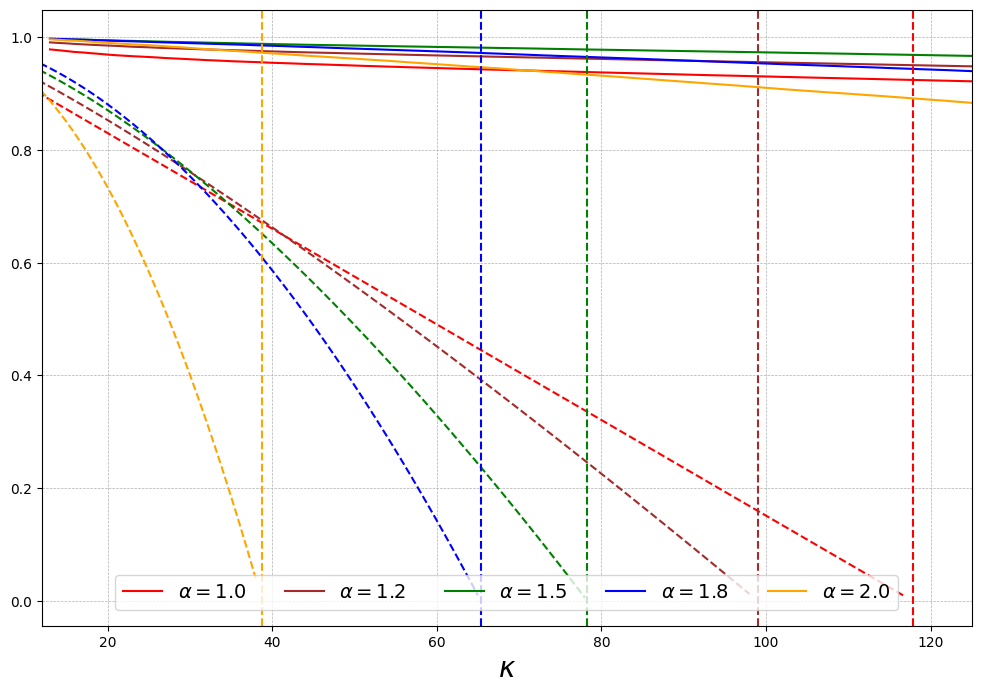} \
\includegraphics[width=8.5cm, height=6.5cm]{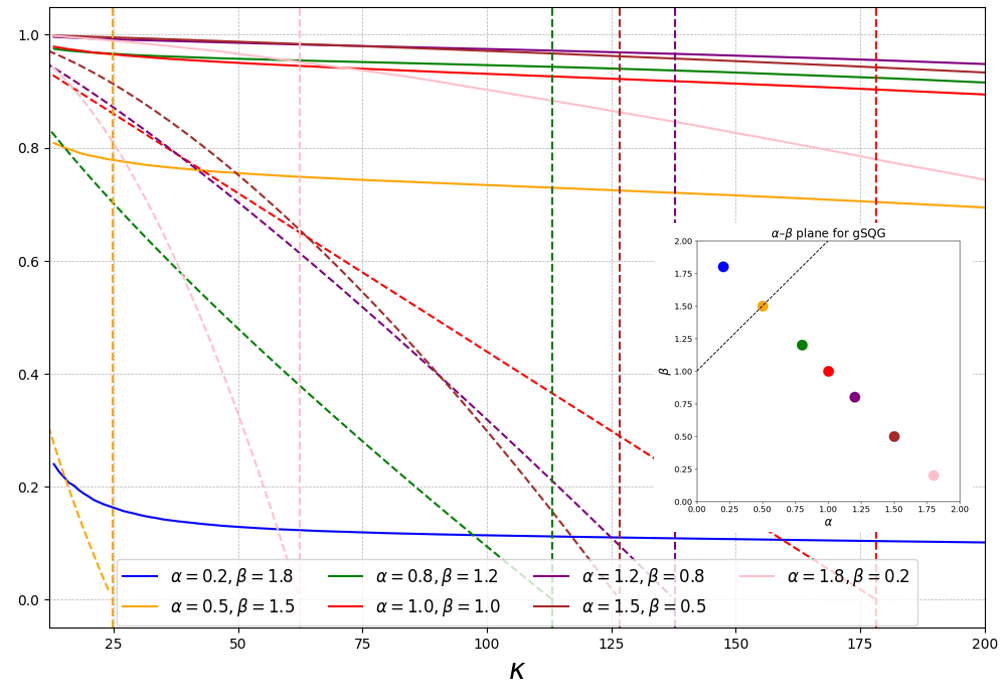}
}
\caption{
Amplified compensated enstrophy flux. Dashed vertical lines indicate $ \kappa_{\sigma}.  $, $N=32768$. Dashed curves indicates $1 - \left( \frac{\kappa}{\kappa_{\sigma}} \right)^{\alpha}$  Left: SQG.  Right: gSQG.
}
\label{fig:ROET_comparison}
\end{figure}

This degradation is most apparent above the ``critical'' line $\beta = \alpha + 1$, 
where the gSQG changes from quasilinear to fully nonlinear.  For those two sample points the value of $\kappa_\sigma$ has fallen below $\bar \kappa$ so Proposition \ref{main_prop1} does not apply.  The results in Figure \ref{fig:ROET_comparison} (right) show that for all samples on or above this line, $\kappa_\sigma < \bar\kappa$, while for those below
$\kappa_\sigma$ exceed $\bar\kappa$ though only slightly.

\section{Relating $\kappa_\sigma$ to $\kappa_\eta$}
\subsection{A rigorous result}

If we quantify the similarity in the spectral power law by specifying a prefactor, $c_{\text{Kr}}$, {up to which the following holds}
\begin{align}\label{withcKr}
\mathcal{E}(\ka) \approx c_{Kr} \eta ^{\frac{2}{3}} \ka_0^{\frac{\beta}{3}} \ka^{\frac{2\beta}{3}-3}, \qquad\
\bar{\ka} \le \ka \le \ka_\eta\;,
\end{align} 
 then we can relate $\ka_{\sigma}$ to $\kappa_{\eta}$.  {We establish such a relation through Proposition \ref{main:prop2} below.} From the plots of the compensated spectra in Figures \ref{Fig:SQGspecs23} and \ref{Fig:gSQGspecs23} it appears that  $c_{\text{Kr}} = \mathcal{O}(1)$.  {Later, in Section \ref{sect:kapeta:G}, we will develop} results that give lower bounds on $\kappa_{\eta}$ in terms of the Grashof number; {these lowers bounds subsequently provide lower bounds on $\ka_\sigma$ through Proposition \ref{main:prop2}, thereby ensuring} a strong enstrophy cascade.
 
\begin{proposition}\label{main:prop2}
Suppose that $\beta>0$, $\ka_\eta \ge 2^{3/\beta}\bar{\ka}$ and ${\delta_1, \de_2} \in (0,1)$ satisfy
    \begin{enumerate}[label=(\roman*)]
        \item $\langle|\te_{\kbar,\ka_\eta}|^2\rangle \ge (1-\delta_1)\langle|\te|^2\rangle$
        \item $\left|\frac{\mathcal{E}(\ka)}{c_{Kr} \eta^{2/3} \ka_0^{\beta/3} \ka^{2\beta/3-3}} - 1\right| \le \delta_2\;,  \qquad
\bar{\ka} \le \ka \le \ka_\eta$ \;.
    \end{enumerate}
Then
    \begin{equation} \label{rel:ksigmaeta}
        \frac{3}{\beta} c_{Kr} \left( \frac{\ka_0}{\bar{\ka}}\right)^{\beta/3} {\left(\frac{1 - \delta_2}{2}\right)} \ka_\sigma^{\al} \le \ka_\eta^{\al} \le \frac{3}{\beta} c_{Kr} \left( \frac{\ka_0}{\bar{\ka}}\right)^{\beta/3} \left(\frac{1+\delta_2}{1-\delta_1}\right) \ka_\sigma^{\al}.
    \end{equation}
\end{proposition}
\begin{proof}
{Let $E: = \frac{1}{L^2} \langle|\te|^2\rangle$. Then}
\begin{align*}
    E & {\le} \frac{1}{(1-\delta_1)L^2} \langle|\te_{\kbar,\ka_\eta}|^2\rangle = {\frac{1}{1-\de_1}\int_{\kapb}^{\kap_\eta}\kap^{2-\be}\Ecal(\kap)\ d\kap}
    \\
    &{\le}  c_{Kr} \eta ^{\frac{2}{3}} \ka_0^{\frac{\beta}{3}} \left(\frac{1+\delta_2}{1-\delta_1}\right) \int_{\bar{\ka}}^{\ka_\eta} \ka^{-1 - \beta/3}\ {d\kap}  \\
    & = \frac{3}{\beta} c_{Kr} \left( \frac{\ka_0}{\bar{\ka}}\right)^{\beta/3} \left(\frac{1+\delta_2}{1-\delta_1}\right) \left[ 1 - \left(\frac{\bar{\ka}}{\ka_\eta}\right)^{\beta/3}\right] \eta ^{2/3} \\
    & \le \frac{3}{\beta} c_{Kr} \left( \frac{\ka_0}{\bar{\ka}}\right)^{\beta/3} \left(\frac{1+\delta_2}{1-\delta_1}\right) \eta ^{2/3}\;.
\end{align*}
    Multiplying both sides by $\eta^{1/3}\ga^{-1} E^{-1}$, {then recalling \eqref{def:eta} and \eqref{def:keta}}, we obtain
    \begin{equation*}
        \ka_{\eta}^{\alpha} = \frac{\eta^{1/3}}{\ga} 
        \le \frac{{3}}{\beta} c_{Kr} \left( \frac{\ka_0}{\bar{\ka}}\right)^{\beta/3} \left(\frac{1+\delta_2}{1-\delta_1}\right)  \frac{\langle|\la^{\frac{\alpha}{2}}\te|^2\rangle}{\langle|\te|^2\rangle} =  \frac{3}{\beta} c_{Kr} \left( \frac{\ka_0}{\bar{\ka}}\right)^{\beta/3} \left(\frac{1+\delta_2}{1-\delta_1}\right) \ka_\sigma^{\al} .
    \end{equation*}
Conversely, we have
    \begin{align*}
        E &\ge \langle|\te_{\kbar,\ka_\eta}|^2\rangle 
        {\ge} c_{Kr} \eta ^{2/3} \ka_0^{\frac{\beta}{3}} (1 - \delta_2) \int_{\bar{\ka}}^{\ka_\eta} \ka^{-1 - \beta/3} d\ka \\
        & = \frac{3}{\beta} c_{Kr} \eta ^{2/3} \left( \frac{\ka_0}{\bar{\ka}}\right)^{\beta/3} (1 - \delta_2) \left[ 1 -  \left(\frac{\bar{\ka}}{\ka_\eta}\right)^{\beta/3}\right]  \\
        & {\ge} \frac{3}{\beta} c_{Kr}  \left( \frac{\ka_0}{\bar{\ka}}\right)^{\beta/3} \left(\frac{1-\delta_2}{2}\right) \eta ^{2/3} \;.
    \end{align*}
Similarly, multiplying both sides by $\eta^{1/3}\ga^{-1} E^{-1}$, we obtain the lower bound in \eqref{rel:ksigmaeta}.

\end{proof}

The estimates in \eqref{rel:ksigmaeta} differ somewhat from the {following analogous result} for the NSE {that was proved} in \cite{Dascaliuc2008}.

\begin{proposition}\label{NSEcase} Consider the NSE ($\alpha=2, \beta=0$) {and suppose that there exists $4\kap_1\leq \kap_\eta$ such that
    \begin{align*}
    \langle |P_{\ka_1}\theta|^2 \rangle \lesssim \langle |Q_{\ka_1}\theta|^2 \rangle.
    \end{align*}
If}
\begin{equation*}
   \mathcal{E}(\ka)\sim\eta^{2/3} \ka^{-3}, \quad \text{for} \quad  \ka_1 \le \ka \lesssim \ka_\eta  \;,\quad 
\end{equation*}    
 then
\begin{equation}\label{etasiglog}
    \ka_\eta^2 \sim \ka_\sigma^2 \ln \frac{\ka_\eta}{\ka_1} \;.
\end{equation}
\end{proposition}

The log term in \eqref{etasiglog} effectively replaces in \eqref{rel:ksigmaeta} the factor $(\ka_0/\overline{\ka})^{\beta/3}/\beta$  
which blows up as $\beta \to 0$.

\subsection{Numerical tests of  \eqref{rel:ksigmaeta}, \eqref{etasiglog}} \label{ksigketacompute}

The assumptions in Proposition~\ref{NSEcase} in the NSE case appear to be well satisfied.   The spectra adhere to the expected power law and the $\ka_{\eta}$ values range from roughly 5 to over 30 times $\bar \kappa$.   In  Figure~\ref{fig:sigmaovereta_NSE_SQG} (left) we plot the quotient of the two sides of \eqref{etasiglog} for $\alpha=2$, varying viscosity and resolutions.  There is less than a 4\% relative change in this quotient over $\gamma \in [10^{-9},10^{-7}]$ in the highest resolution tested $N=16384$.

\begin{figure}[H]
\centerline{
\includegraphics[width=8.5cm, height=6.5cm]{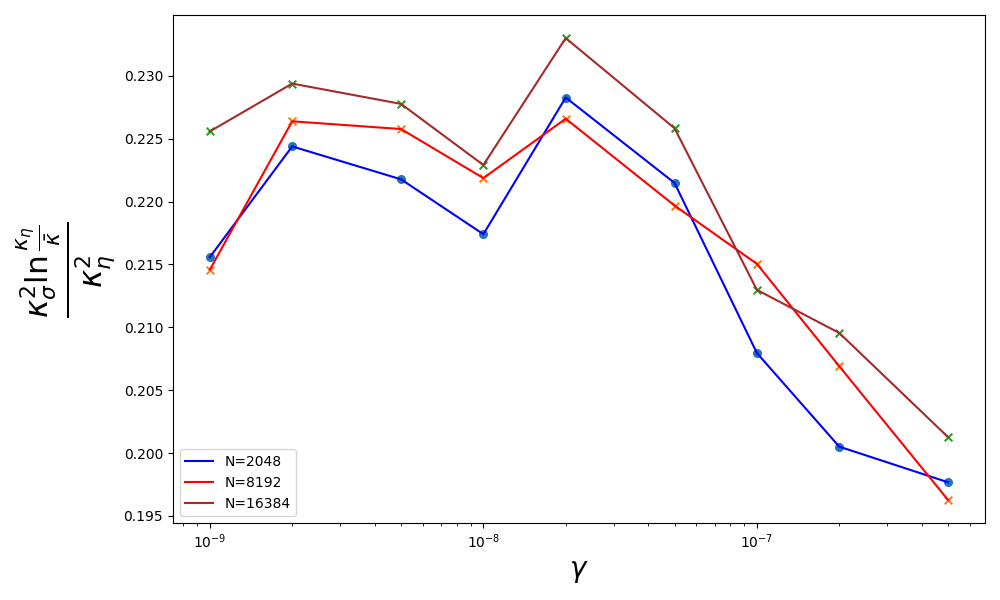}
\includegraphics[width=8.5cm, height=6.5cm]{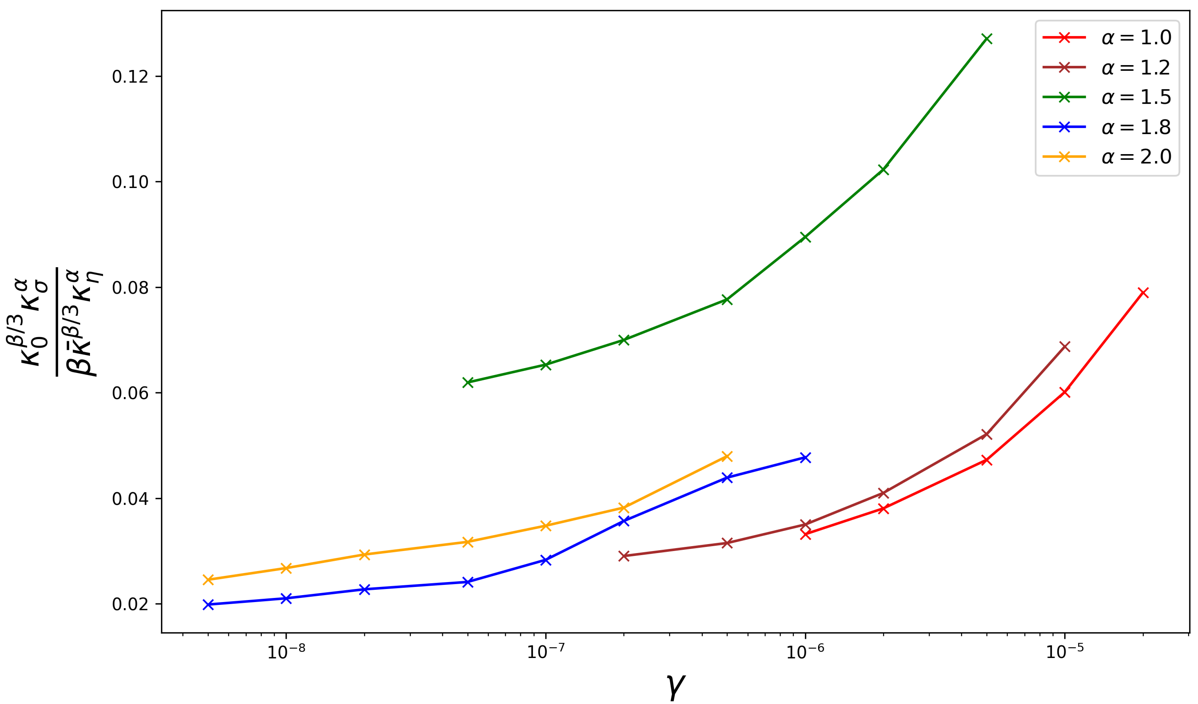}}
\caption{Left: test of \eqref{etasiglog} (NSE). Right: test of  \eqref{rel:ksigmaeta} (SQG), $N=16384$
}
\label{fig:sigmaovereta_NSE_SQG}
\end{figure}

In the case of the SQG, the spectra sampled for those values of $\alpha\ge 1$ in Figure \ref{Fig:SQGspec} match the power law rather faithfully, though not as well as for the NSE.  In particular, we note a hint of inflation  just before the dissipation fall-off for $\alpha=1$.  The gap between $\bar \ka$ and $\ka_\eta$ is even wider than in the NSE case.   Figure~\ref{fig:sigmaovereta_NSE_SQG} (right) shows the relevant quotient for \eqref{rel:ksigmaeta} to decrease as $\gamma$ decreases over the ranges where we computed $\ka_\eta$.  However the convexity of each plot suggests there is a leveling off at yet smaller viscosity values. 

In the case of the gSQG the computational results  are mixed.
In the southeast corner of the inset in Figure \ref{fig:sigmaeta_gSQG_NW} the quotient is relatively flat for larger $\alpha$, smaller $\beta$ and steepening as the critical SQG case (in red) is approached. 
It is in the northwest corner that we see the pronounced bump in the spectrum near the dissipation fall-off (see Figure~\ref{Fig:gSQGspectra}). Proposition \ref{main:prop2} then should not apply. This inflation of energy at smaller scales results in larger $\ka_\eta$, which is consistent with the steeper plots of the quotient. 
In Figure \ref{fig:sigmaeta_gSQG_NW} (right) we plot the quotient for 9 sample points within this quadrant. The range over which we vary viscosity is again taken to be that which enables the computation of $\ka_\eta$.  The behavior of the quotient seems disassociated from critical line $\beta=\alpha+1$.   

\setlength{\intextsep}{1pt}
\begin{figure}[H]
\centerline{
\includegraphics[width=8.5cm, height=6.15cm]{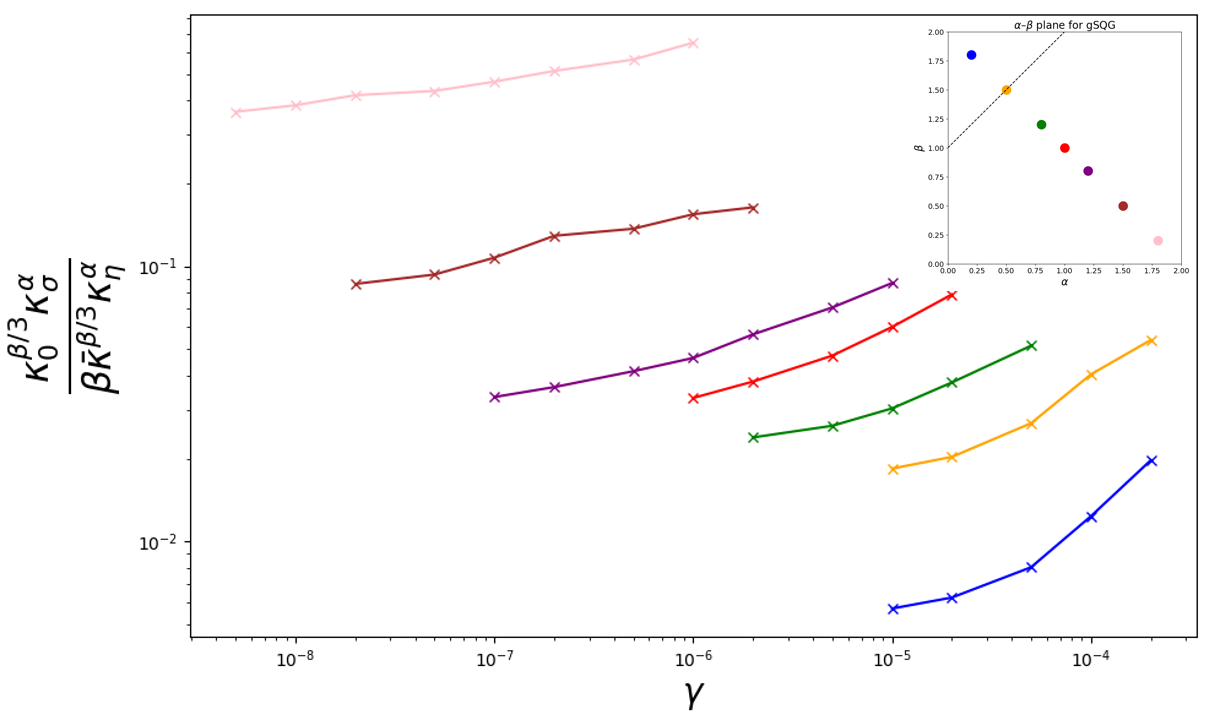} \
\includegraphics[width=8.5cm, height=6.5cm]{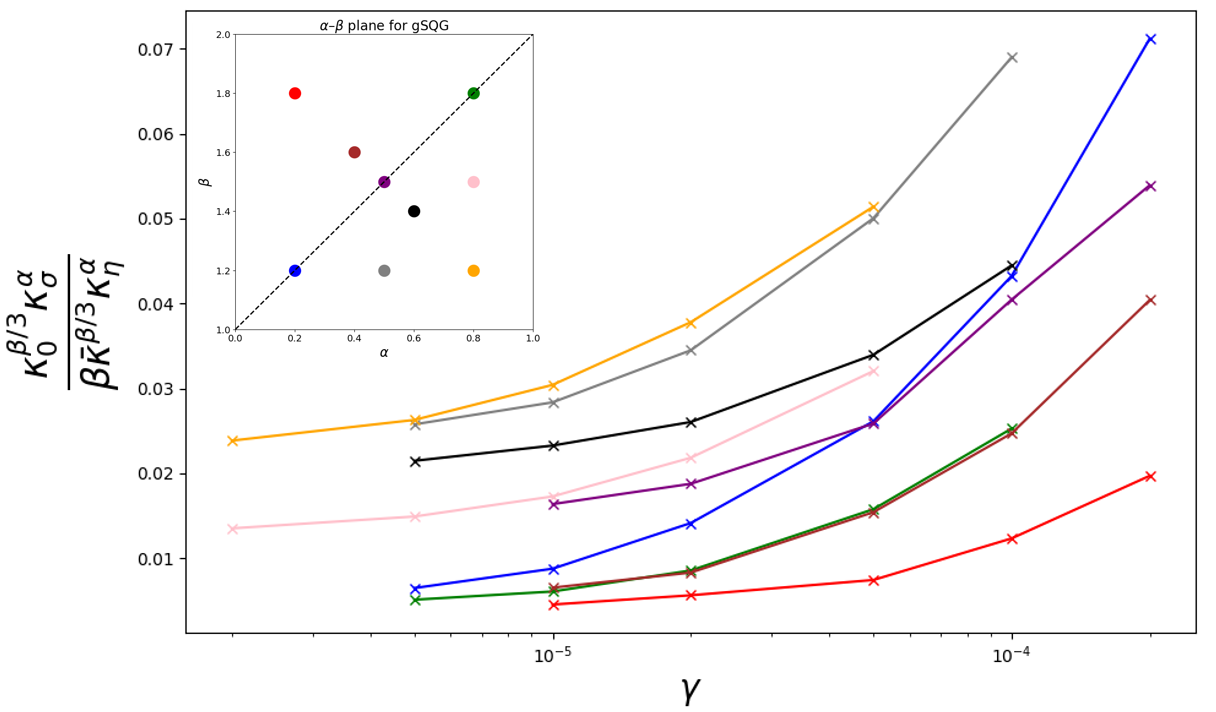}
}
\caption{
Test of \eqref{rel:ksigmaeta}. Left: gSQG. Right: NW Quadrant, $N=16384$
}
\label{fig:sigmaeta_gSQG_NW}
\end{figure}

\section{dissipation law}
Kolmogorov's dissipation law for the 3D NSE
    \begin{equation} \label{rel:etau_3DNSE}
        \epsilon \sim \frac{U_{\text{3D}}^3}{L}\;,
    \end{equation}
where 
$$\epsilon=\frac{\gamma}{L^3}\langle |\theta|^2\rangle\qquad \text{and} \qquad U_{\text{3D}}=\ka_0^{3/2} \langle |u|^2 \rangle^{1/2}$$
are the {\it energy dissipation rate} and {\it root mean square velocity}.  The 2D analogue of Kolmogorov's 3D dissipation law is
    \begin{equation} \label{rel:etau_NSE}
        \eta \sim \frac{U_{\text{2D}}^3}{L^3}\qquad \text{with} \qquad U_{\text{2D}}={\frac{1}{L}}
        \langle |u|^2 \rangle^{1/2}\;.
    \end{equation}
If this law were to extend to the gSQG, it would take the form 
    \begin{equation} \label{disslawgSQG}
    \eta \sim \frac{U^3}{L^3}\qquad \text{but with} \qquad     U ={L^{\frac{\be}2-1}}
    \langle|\la^{\frac{\beta-2}{2}}\te|^2\rangle^{1/2} .
    \end{equation}

 \subsection{Analytical support}\label{rigdisslaw}
   
 The upper bound  $\epsilon \lesssim {U_{\text{3D}}^3}/{L}$ was proved in \cite{constantin1994shear} for shear and channel flow and in \cite{foias1997turbulence} for the periodic NSE.  {We show that the} analogous upper bound for gSQG holds provided that the Grashof number is sufficiently large. 
\begin{theorem} \label{thm:etau}
 There exists $\Gb_1$, depending only on $\kap_0,\bkap,\al,\be$, such that if $ G \gtrsim \Gb_1 $, then
    \begin{equation} \label{rel:etau_gSQG}
        \eta \lesssim \frac{U^3}{L^3}.
    \end{equation}    
    \end{theorem}


{Our} proof of Theorem \ref{thm:etau} involves {the use of commutator operators that allow us to obtain more precise estimates on dyadic triad interactions that arise through localizing the nonlinear term in the enstrophy balance. This localization is facilitated by the Littlewood-Paley decomposition, which we} provide a brief introduction to 
in the appendix. {For efficiency of the present goal of proving Theorem \ref{thm:etau}, we simply state the commutator estimates invoked its proof (Lemma \ref{JKM22lemma} and Lemma \ref{lem:commutator4} below). We remark that} Lemma \ref{lem:commutator4} for the setting of $\Omega=\mathbb{R}^2$ and $i=j$ was proved in \cite{JKM2022}. The proof in the current setting of $\Omega=\mathbb{T}^2$ and $i\sim j$ is almost identical and therefore we omit it to avoid redundancy. The proof of Lemma \ref{JKM22lemma} follows an argument analogous to Lemma 3.1 in \cite{JKM2022}, but with enough modifications {that justify supplying additional details for (see Appendix \ref{sect:appendix})}. 

\begin{lemma}
\label{JKM22lemma}
    For $s\in (0,1)$, let $\rho \in [0,s]$. Let $f_1,f_2,f_3\in H$. Given $i, j\in \mathbb{Z}$ such that $|i-j|\le k$ for a positive integer $k$, suppose that $\text{supp} \ \hat{f}_2\subset \mathcal{A}_i$ and $\text{supp} \ \hat{f}_3\subset \mathcal{A}_j$. Then there exists a constant $C>0$, depending on $s,k$ and $\rho$ such that
    \[|\lbn [\la^{-s}\nabla,f_1]f_2,f_3\rbn|\le C\Sob{\widehat{\la^{1-s+\rho} f_1}}{\ell^1(\ZZ^2)}|\la^{-\rho}f_2||f_3|.\]
\end{lemma}
    \begin{lemma}\label{lem:commutator4}
Let $s \in [0,1)$, $\rho \in \mathbb{R}$. Let $f_1,f_2,f_3\in H$. Given $i, j\in \mathbb{Z}$ such that $|i-j|\le k$ for a positive integer $k$, suppose that $\supp\hat{f}_3\subset{\Acal}_j$. Then there exists a sequence $\{c_i\}\in\ell^2(\ZZ)$ such that $\Sob{\{c_i\}}{\ell^2(\mathbb{Z})}\leq1$ and
    \begin{align*}
        |([\Lam^{s+\rho+1}\Delta_i,f_1]f_2,f_3)|\leq Cc_{i}
              (\Sob{\Lam \widehat{f}_1}{\ell^1(\ZZ^2)}+|\Lam^2f_1|)|\Lam^sf_2||\Lam^\rho f_3|,
    \end{align*}     
for some constant $C>0$, depending only on $s,\rho,k$. 
\end{lemma}

With these in hand, we now prove Theorem \ref{thm:etau}. 
\begin{proof}[Proof of Theorem \ref{thm:etau}]
    First, multiply equation \eqref{maineq} by $\te$, integrate over $\mathbb{T}^2$ and average in time to obtain
    \begin{equation} \label{eq:eq1_theta}
        \ga \langle |\la^{\frac{\alpha}{2}}\te|^2 \rangle = \langle (g, \te) \rangle .   
    \end{equation}
    From \eqref{def:eta}, \eqref{disslawgSQG}, \eqref{eq:eq1_theta}, and  Jensen's inequality, we have
    \begin{equation} \label{ineq:etaineq}
        \eta = \frac{1}{L^2} \langle (g, \te) \rangle \le \frac{1}{L^2} \langle|\la^{\frac{\beta-2}{2}}\te|\rangle |\la^{\frac{2-\beta}{2}}g| 
        \le \frac{U}{L^2} \bar{\ka}^{\frac{2+\alpha-\beta}{2}}{L^{1-\frac{\beta}{2}}} |\la^{-\frac{\alpha}{2}}g| .
    \end{equation}
    We then take the {$H$-scalar} product of \eqref{maineq} by $\la^{-\alpha}g$ and average in time to get
    \begin{equation*} 
     \ga \langle (\la^{\alpha}\te, \la^{-\alpha}g) \rangle + \langle (u\cdotp\nabla \te, \la^{-\alpha}g) \rangle = \langle (g, \la^{-\alpha}g) \rangle .
    \end{equation*}
Upon dividing by $L^2$ and {using self-adjointness of} the fractional laplacian operator, we get
    \begin{equation} \label{eq:eq1_la_nalpha_g}
        \frac{\ga}{L^2} \langle (\te,g) \rangle + \frac{1}{L^2}\langle (u\cdotp\nabla \te, \la^{-\alpha}g) \rangle  = \frac{1}{L^2} |\la^{-\frac{\alpha}{2}} g|^2  .
    \end{equation}
For $0<\be<2$ we have
\begin{align*}
    \ka_0^\be\lbn u\cdot \nabla \te, \la^{-\alpha}g \rbn=&\lbn \nabla^{\perp}\la^{\beta-2}\te\cdot \nabla \te, \la^{-\alpha}g \rbn\\
    =&\lbn \nabla^{\perp}\la^{\beta-2}\te\cdot \nabla \te, \la^{-\alpha}g \rbn-\underbrace{\lbn \nabla^{\perp}\la^{\frac{\beta}{2}-1}\te\cdot \nabla \la^{\frac{\beta}{2}-1}\te,\la^{-\alpha}g \rbn}_{=0}\\
    =&-\lbn \nabla^{\perp}\la^{\be-2}\te\cdot \nabla\la^{-\al}g,\te\rbn+\lbn \nabla^{\perp}\la^{\frac{\be}{2}-1}\te \cdot \nabla \la^{-\al}g,\la^{\frac{\be}{2}-1}\te \rbn\\
    =&-\lbn \la^{\frac{\be}{2}-1}(\nabla^{\perp}\la^{\frac{\be}{2}-1}\te)\cdot \nabla \la^{-\al}g,\te\rbn+\lbn \la^{\frac{\be}{2}-1}(\nabla^{\perp}\la^{\frac{\be}{2}-1}\te\cdot \nabla \la^{-\al}g),\te\rbn\\
    =&\lbn [\la^{\frac{\be}{2}-1}\nabla^{\perp}\cdot, \nabla \la^{-\al}g]\la^{\frac{\be}{2}-1}\te,\te\rbn\\
    =& \sum_{|j-k|\le 1} \lbn \lpj ([\la^{\frac{\be}{2}-1}\nabla^{\perp}\cdot, \nabla \la^{-\al}g]\la^{\frac{\be}{2}-1}\te),\lpk \te\rbn.
\end{align*}
We now use the decomposition
\begin{align*}
      \ka_0^\be\lbn u\cdot \nabla \te, \la^{-\alpha}g \rbn=&\sum_{|j-k|\le 1}\left \{ \lbn [\nabla^{\perp}\la^{\frac{\be}{2}-1}\cdot,\nabla\la^{-\al}g]\lpj \la^{\frac{\be}{2}-1}\te,\lpk \te\rbn \right.\\& \hspace{4 em}  \left.-\lbn [\lpj,\nabla\la^{-\al}g]\cdot \nabla^{\perp} \la^{\be-2}\te,\lpk \te\rbn \right.\\& \hspace{4 em}  \left.  +\lbn [\lpj,\nabla\la^{-\al}g]\cdot \nabla^{\perp} \la^{\frac{\be}{2}-1}\te,\lpk \la^{\frac{\be}{2}-1}\te\rbn\right\}\\
      =&\sum_{|j-k|\le 1} I+II+III
\end{align*}
which can be verified by expanding each of $I$, $II$ and $III$.

Since $\be>0$, we may apply Lemma \ref{JKM22lemma} with $s=\rho=1-\frac{\be}{2}$ and then Bernstein's inequality to find
\begin{align*}
    |I|\le& C\|\widehat{\nabla \la^{1-\al}g}\|_{\ell^1(\ZZ^2)}
    |\lpj\la^{\be-2}\te||\lpk\te|\\
    \le&C\|\widehat{\nabla \la^{1-\al}g}\|_{\ell^1(\ZZ^2)}|\lpj\la^{\frac{\be}{2}-1}\te||\lpk \la^{\frac{\be}{2}-1}\te|.
\end{align*}
Next note that 
\[II=\lbn [\lpj,\nabla\la^{-\al}g]\cdot \nabla^{\perp} \la^{\be-2}\te,\lpk \te\rbn =-\lbn [\lpj,\nabla\la^{-\al}g] \la^{\be-2}\te \cdot,\nabla^{\perp}\lpk \te\rbn. \]
Applying Lemma \ref{lem:commutator4} with $s=1-\frac{\be}{2},\, \rho=\frac{\be}{2}-2$, we get
\begin{align*}
    |II|\le Cc_j \left(\Sob{ \widehat{\nabla \la^{1-\al}g}}{\ell^1(\ZZ^2)}+|\nabla \la^{2-\al}g|\right)|\la^{\frac{\be}{2}-1}\te||\lpk \la^{\frac{\be}{2}-1}\te|,
\end{align*}
{for some $c_j \in \ell^{2}(\mathbb{Z})$}. We next apply Lemma \ref{lem:commutator4} with $s=0,\, \rho=-1$ on
\[III=-\lbn [\lpj,\nabla\la^{-\al}g]\cdot \nabla^{\perp} \la^{\frac{\be}{2}-1}\te,\lpk \la^{\frac{\be}{2}-1}\te\rbn =\lbn [\lpj,\nabla\la^{-\al}g] \la^{\frac{\be}{2}-1}\te \cdot,\nabla^{\perp}\lpk \la^{\frac{\be}{2}-1}\te\rbn,\]
to obtain an identical upper bound on $|III|$.

From $\|\hat f\|_{\ell^1(\ZZ^2)} \le (1/\sqrt{\pi})|f|^{1/2}|\la^2 f|^{1/2}$ we have $\|\widehat{\nabla \la^{1-\alpha} g}\|_{\ell^1(\ZZ^2)} \le \bar{\ka}^{3-\alpha/2}|\la^{-\alpha/2}g|$. {We then treat the sum of the resulting terms from $I$ over $|j-k|\leq1$} using \eqref{T:Bernstein} with $\sigma=0$, {and the Cauchy-Schwarz inequality}. {Similarly, we bound the sum of the resulting terms from $II$ and $III$ over $|j-k|\leq1$ using the Cauchy-Schwarz inequality and the fact that $c_j\in\ell^2(\ZZ)$. Finally, upon taking time-averages}, we conclude 
\begin{align} \label{ineq:bterm}
 \langle (u\cdotp\nabla \te, \la^{-\alpha}g) \rangle \le  C \bar{\ka}^{3-\frac{\alpha}{2}}{L^{\beta}} |\la^{-\frac{\alpha}{2}}g| \langle |\la^{\frac{\beta-2}{2}}\te|^2 \rangle 
\end{align}
for {some} dimensionless constant $C$.

 Now \eqref{ineq:etaineq}, \eqref{eq:eq1_la_nalpha_g} and \eqref{ineq:bterm} gives us
    \begin{equation*} 
    \begin{split}
    \frac{1}{L^2} |\la^{-\frac{\alpha}{2}} g|^2 
    & \le \ga \eta + \frac{C}{L^2} \bar{\ka}^{3-\frac{\alpha}{2}}{L^{\beta}} |\la^{-\frac{\alpha}{2}}g| \langle |\la^{\frac{\beta-2}{2}}\te|^2 \rangle 
    \\
    & \le \ga \frac{U}{L^2}\bar{\ka}^{\frac{2+\alpha-\beta}{2}} {L^{1-\frac{\beta}2}} |\la^{-\frac{\alpha}{2}}g| + \frac{C}{L^2} \bar{\ka}^{3-\frac{\alpha}{2}}{L^{\beta}} |\la^{-\frac{\alpha}{2}}g| \langle |\la^{\frac{\beta-2}{2}}\te|^2 \rangle .
    \end{split} 
    \end{equation*}
{The definition of $U$ from \eqref{disslawgSQG} then implies}
    \begin{equation} \label{ineq:g}
        \frac{1}{L^2} |\la^{-\frac{\alpha}{2}} g| \le \ga \frac{U}{L^2}\bar{\ka}^{\frac{2+\alpha-\beta}{2}}{L^{1-\frac{\beta}2}} + C \bar{\ka}^{3-\frac{\alpha}{2}} U^2.
    \end{equation}
Suppose that $U < \ga \bar{\ka}^{\alpha+\frac{\beta}{2}-2} {L^{\frac{\beta}2-1}}$, then apply \eqref{def:G} and \eqref{ineq:g} to get
    \begin{equation*}
     \frac{1}{L^2} |\la^{-\frac{\alpha}{2}} g|  <  \frac{\ga^2}{L^2}\bar{\ka}^{\frac{3\alpha-2}{2}} + C{\frac{\gam^2}{L^2}\bkap^{\frac{3\al-2+2\be}2}L^\be=\frac{\gam^2}{L^2}\bkap^{\frac{3\al-2}{2}}\left(1+C\bkap^\be L^\be\right)}.
    \end{equation*}
{It follows that}
    \begin{align}\label{ineq:g1}
  |\la^{-\frac{\alpha}{2}} g| < \ga^2 \bkap^{\frac{3\al-2}{2}}\left[1 + {\frac{C}{(2\pi)^\be}}\left(\frac{\bar{\ka}}{\kap_0}\right)^{\beta}\right]
    \end{align}
{Now, define}
    \begin{align}\notag
        \Gb_1\overset{\text{def}}{=}\left(\frac{\bar{\ka}}{\kap_0}\right)^{\frac{3\alpha-2}{2}}\left[1 + C\left(\frac{\bar{\ka}}{\kap_0}\right)^{\beta}\right].
    \end{align}
{It follows from \eqref{ineq:g1} that}
   \begin{equation*}
    G  < \Gb_1
    \end{equation*}            
{In other words, we have shown that if $G \ge \underline{G}$, then $U \ge \ga \bar{\ka}^{\alpha+\frac{\beta}{2}-2} L^{\frac{\beta}2-1}$

Therefore, if $G\geq \Gb_1$, then upon returning to  \eqref{ineq:g} we deduce}
    \begin{equation} \label{ineq:g2}
        \frac{1}{L^2} |\la^{-\frac{\alpha}{2}} g| \le U^2 \bar{\ka}^{3-\frac{\alpha}{2}} \left[{\frac{1}{(\bar{\ka}L)^{\beta}}}+C\right].
    \end{equation}
{Finally, by applying \eqref{ineq:g2} in \eqref{ineq:etaineq} we deduce}
    \begin{equation*}
    \begin{split}
        \eta 
        \leq\frac{U}{L^2} \bar{\ka}^{\frac{2+\alpha-\beta}{2}}{L^{1-\frac{\beta}{2}}} |\la^{-\frac{\alpha}{2}}g|
        \leq {\frac{U^3}{L^3}\left(\bkap L\right)^{4-\frac{\be}2}\left[\frac{1}{(\bkap L)^\be}+C\right]},
    \end{split}
    \end{equation*}
{as desired.}

For the case $\be=0$, we can instead use the following:
\begin{align}
    (u\cdotp\nabla\tht,\Lam^{-\al}g)&=(\nabla^\perp\psi\cdotp\nabla\De\psi,\Lam^{-\al}g)=-(\nabla^\perp\bdy_j\psi\cdotp\nabla\bdy_j\psi,\Lam^{-\al}g)-(\nabla^\perp\psi\cdotp\nabla\bdy_j\psi,\Lam^{-\al}\bdy_jg)\notag
    \\
    &=-(\nabla^\perp\psi\cdotp\nabla\bdy_j\psi,\Lam^{-\al}\bdy_jg)=(\nabla^\perp\psi\cdotp\nabla\Lam^{-\al}\bdy_jg,\bdy_j\psi).\notag
\end{align}
Then
    \begin{align}
        |\lb(u\cdotp\nabla\tht,\Lam^{-\al}g)\rb|&\leq \lb|u|^2\rb|\Lam^{2-\al}g|_\infty\notag
    \end{align}
Note that when $\be=0$, $|u|^2=|\Lam^{\be-1}\tht|^2=|\Lam^{\frac{\be-2}{2}}\tht|^2$. This eventually yields the analog of \eqref{ineq:bterm}. The analysis then proceeds as above.
\end{proof}


It was shown in \cite{Dascaliuc2009} that if the Kolmogorov $5/3$ spectrum holds, then \eqref{rel:etau_3DNSE} holds for the 3D NSE.  Similarly, for the 2D NSE it was shown in \cite{Dascaliuc2008} that if $\mathcal{E}(\ka) \sim \eta^{2/3}\ka^{-3}$ holds, then 
\begin{equation}\label{uptolog}
 \frac{U_{\text{2D}}^3}{L^3} \lesssim (\log G)^{15/4} \eta 
\end{equation}
so that \eqref{rel:etau_NSE} holds up to a logarithm in $\gamma$ for fixed $L$ and force $g$.  That we should be unable to extend this result to the gSQG is suggested by the numerical tests which follow.

\subsection{Numerical tests of \eqref{rel:etau_NSE}, \eqref{rel:etau_gSQG}}
We test these relations by plotting the quotient of the two sides of the relation over a range of $\gamma$ values.

Figure \ref{fig:etau_side_by_side} (right - plotted in gray) confirms that  \eqref{rel:etau_NSE} holds for the NSE as the quotient remains nearly constant over several decades of viscosity values.  The slight increase as $\gamma$ decreases is consistent with the log term in \eqref{uptolog}.  The quotients plotted for the SQG and gSQG in Figure \ref{fig:etau_side_by_side} (left and right, resp.) show marked increase as $\gamma$ decreases, which is consistent however with  \eqref{rel:etau_gSQG}.

\setlength{\intextsep}{1pt}
\begin{figure}[H]
\centerline{
\includegraphics[width=8.5cm, height=6.5cm]{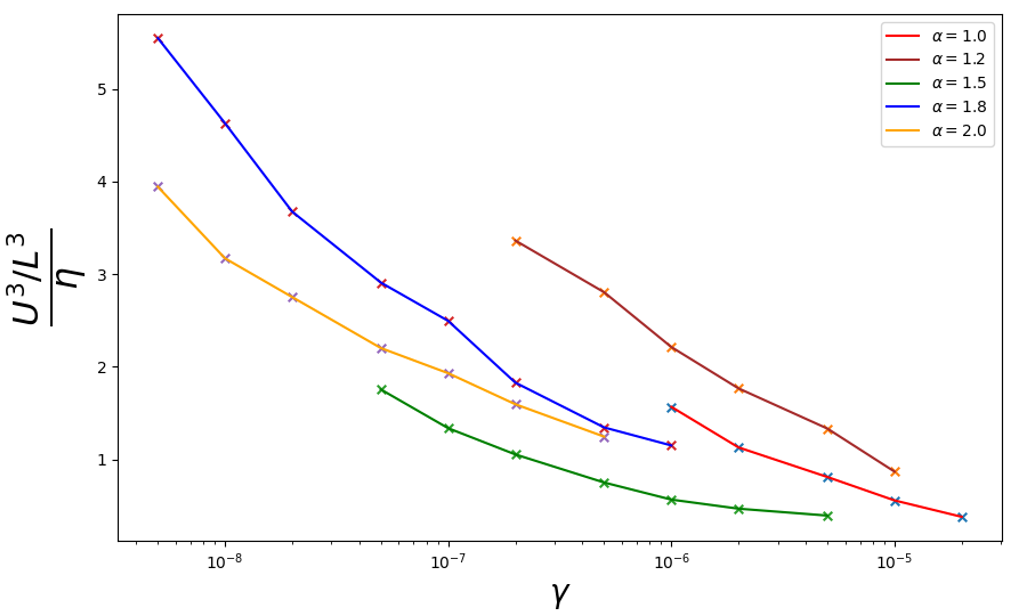} \
\includegraphics[width=8.5cm, height=6.5cm]{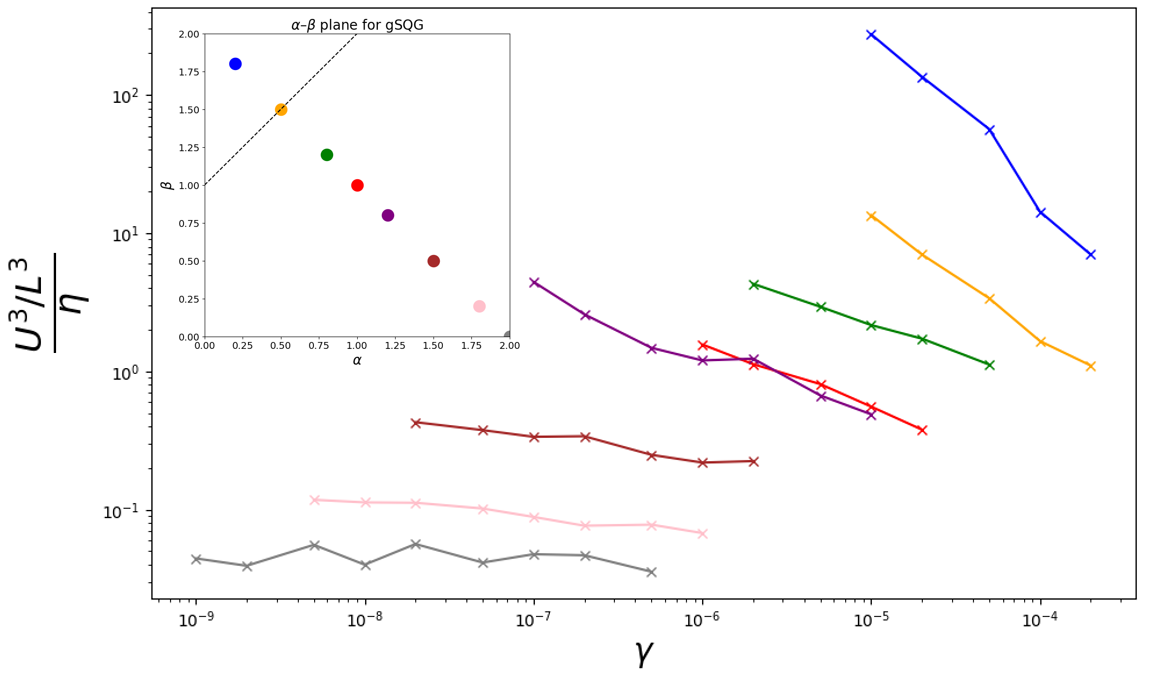}
}
\caption{
Test of \eqref{rel:etau_NSE}, \eqref{rel:etau_gSQG}. Left: SQG. Right: gSQG, $N=16384$
}
\label{fig:etau_side_by_side}
\end{figure}

\section{Relating $\ka_\eta$ to $G$}\label{sect:kapeta:G}
\subsection{Analytical Support}
 We start with upper and lower bounds on $\ka_\eta$
that {essentially} hold {unconditionally in a context consistent with turbulence}.  {These bounds in the} NSE case {were} proved in \cite{FMT93}. Our proof for the gSQG {case exploits}
the estimate \eqref{ineq:bterm} for the nonlinear term {from} subsection \ref{rigdisslaw}.

\begin{proposition} \label{prop:keta_G}
{Suppose that $\kap_\eta\geq\kap_0$. Then}
    \begin{equation}\label{bounds:kaeta}
        G^{\frac{1}{3\alpha}} \lesssim \frac{\kappa_{\eta}}{\ka_0} \le G^{\frac{2}{3\alpha}}\;.
    \end{equation}
\end{proposition}

\begin{proof}
{Recall $\eta$ as defined in \eqref{def:eta}}. We revisit \eqref{ineq:etaineq} to {alternatively estimate} 
    \begin{equation} \label{ineq:etaineq2}
        \eta = \frac{1}{L^2} \langle (g, \te) \rangle \le \frac{1}{L^2} |\langle|\la^{\frac{\alpha}{2}}\te|\rangle| |\la^{-\frac{\alpha}{2}}g| 
        = \frac{1}{L\ga^{1/2}} \eta^{\frac{1}{2}} |\la^{-\frac{\alpha}{2}}g| .
    \end{equation}
{In particular}
     \begin{equation} \label{ineq:etahalf}
        \eta \le \frac{1}{L^2 \ga} |\la^{-\frac{\alpha}{2}}g|^2.
    \end{equation}
{The claimed upper bound then follows from \eqref{def:keta} and \eqref{def:G}.}
    
{For the lower bound, observe that since $\alpha \ge \beta-2$, we} have from \eqref{ineq:bterm} {that}, 
    \begin{align} \label{ineq:bterm2}
    \langle (u \cdotp\nabla \te, \la^{-\alpha}g) \rangle  & \le C \bar{\ka}^{3-\frac{\alpha}{2}}\ka_0^{-\beta} |\la^{-\frac{\alpha}{2}}g| \langle |\la^{\frac{\beta-2}{2}}\te|^2 \rangle \notag
    \\
    & \le C \bar{\ka}^{3-\frac{\alpha}{2}} \ka_0^{-2-\alpha} |\la^{-\frac{\alpha}{2}}g| \langle |\la^{\frac{\alpha}{2}}\te|^2 \rangle \notag
    \\
    &\leq C \bar{\ka}^{3-\frac{\alpha}{2}} L^{2+\alpha} |\la^{-\frac{\alpha}{2}}g| \langle |\la^{\frac{\alpha}{2}}\te|^2 \rangle.   
    \end{align}   
Applying \eqref{ineq:etaineq2} and \eqref{ineq:bterm2} to \eqref{eq:eq1_la_nalpha_g}, we get
    \begin{equation*} 
         \frac{1}{L^2} |\la^{-\frac{\alpha}{2}} g|^2 \le \frac{\ga^{1/2}}{L} \eta^{\frac{1}{2}} |\la^{-\frac{\alpha}{2}}g| + \frac{C}{L^2} \bar{\ka}^{3-\frac{\alpha}{2}} \ka_0^{-2-\alpha} |\la^{-\frac{\alpha}{2}}g| \langle |\la^{\frac{\alpha}{2}}\te|^2 \rangle.
    \end{equation*}
In particular
    \begin{align}\notag
        \frac{|\la^{-\frac{\alpha}{2}}g|}{\ga^2 \kappa_0^{\frac{3\alpha-2}{2}}} \le \frac{L}{\kappa_0^{\frac{3\alpha-2}{2}}} \left(\frac{\eta}{\ga^3}\right)^{\frac{1}{2}} + \frac{CL^2}{\kappa_0^{\frac{3\alpha-2}{2}}}\bar{\ka}^{3-\frac{\alpha}{2}} \ka_0^{-2-\alpha}\left(\frac{\eta}{\ga^3}\right)
    \end{align}
Therefore
    \begin{align}\label{ineq:g3}
        G \le C_1 \left(\frac{\ka_\eta}{\ka_0}\right)^{\frac{3\alpha}{2}} + C_2 \left(\frac{\ka_\eta}{\ka_0}\right)^{3\alpha}\leq \frac{1}2G+\left(\frac{C_1^2}{2G}+C_2\right)\left(\frac{\ka_\eta}{\ka_0}\right)^{3\alpha},
    \end{align}
{for some non-dimensional constants $C_1, C_2$. It is now clear that if $\kap_\eta\geq\kap_0$, then $G^{\frac{1}{3\alpha}} \lesssim \frac{\ka_\eta}{\ka_0}$.}  
\end{proof}

{
\begin{remark}
Observe that if $\kap_\eta<\kap_0$, then \eqref{ineq:g3} implies that $G\leq (C_1+C_2)(\kap_\eta/\kap_0)^{3\al/2}$. Thus, in this case
    \begin{align}\notag
        G^{\frac{2}{3\al}}\sim\frac{\kap_\eta}{\kap_0}.
    \end{align}
\end{remark}
}

\subsection{Refinement of analytical bounds under turbulence}
In the case of the NSE, the upper and lower bounds on $\ka_\sigma$ and $\ka_\eta$ can be sharpened to the same power in $G$ up to a logarithm, provided \eqref{etasiglog} holds.  
The following {is}  proved in \cite{Dascaliuc2008}.
\begin{theorem}\label{DFJketaksig}
    If \eqref{etasiglog} holds and $G\gtrsim(\kbar/\ka_0)^2 $, then
for the 2D NSE
\begin{equation}\label{ksigNSE}
    \left(\frac{\ka_0}{\kbar}\right)^{5/4} \frac{G^{1/4}}{(\ln G)^{3/2}}  \lesssim \frac{\ka_\sigma}{\ka_0} \lesssim 
\left(\frac{\kbar}{\ka_0}\right)^{5/4} {G^{1/4}}{(\ln G)^{3/8}}
\end{equation}
\begin{equation}\label{ketaNSE}
\left(\frac{\ka_0}{\kbar}\right)^{1/4} \frac{G^{1/4}}{(\ln G)^{3/2}}  \lesssim \frac{\ka_\eta}{\ka_0} \lesssim 
\left(\frac{\kbar}{\ka_0}\right)^{1/4} {G^{1/4}}{(\ln G)^{1/8}}.
    \end{equation}
\end{theorem}

Here, we have an analogue for the gSQG.

\begin{theorem} \label{main:thm_G_wavenumbers}
   {There exists $\Gb_2$, depending only on $\kap_0,\bkap,\al$, such that if $ G \gtrsim \Gb_2 $ and \eqref{rel:ksigmaeta} holds}, then
    \begin{align} \label{G:ksigma}
      \beta^{\frac{3}{\al}} \left(\frac{\ka_0}{\bar{\ka}}\right)^{\frac{18-\al-4\beta}{4\al}} G^{1/2\alpha} &\lesssim \frac{\ka_\sigma}{\ka_0} \lesssim \left(\frac{1}{\beta}\right)^{\frac{3}{4\al}} \left(\frac{\bar{\ka}}{\ka_0}\right)^{\frac{12+\al-\beta}{4\al}} G^{1/2\alpha}
        \\
        \beta^{\frac{1}{2\al}} \left(\frac{\ka_0}{\bar{\ka}}\right)^{\frac{18-3\al-2\beta}{12\alpha}} G^{1/2\alpha} &\lesssim \frac{\ka_\eta}{\ka_0} \lesssim \left(\frac{1}{\beta}\right)^{\frac{1}{4\al}} \left(\frac{\bar{\ka}}{\ka_0}\right)^{\frac{3\alpha-\beta}{12\alpha}} G^{1/2\alpha}. \label{G:keta}
    \end{align}
\end{theorem}
Note that by Proposition \ref{prop:keta_G} we can guarantee $\ka_\eta$ is large by taking $G$ sufficiently large. {Thus,} by Proposition \ref{main:prop2} (resp. Proposition \ref{NSEcase}), the technical assumption \eqref{rel:ksigmaeta} (resp. \eqref{etasiglog}) can be replaced by the commonly observed spectral assumption \eqref{withcKr}.
\begin{proof}
We first show that
\begin{equation} \label{thm4:q1}
    \left(\frac{\ka_0}{\bar{\ka}}\right)^{3-\frac{\al}{2}} G \lesssim \frac{\langle |\te|^2 \rangle}{\ga^2 \ka_0^{2\al-2}} \lesssim \left(\frac{1}{\beta}\right)^{\frac{3}{2}} \left(\frac{\bar{\ka}}{\ka_0}\right)^{\frac{\al-\beta}{2}} G 
\end{equation}
and
\begin{equation} \label{thm4:q2}
    \beta^{3/2} \left(\frac{\ka_0}{\bar{\ka}}\right)^{\frac{18-3\al-2\beta}{4}} G^{3/2} \lesssim \frac{\langle |\la^{\frac{\al}{2}}\te|^2 \rangle}{\ga^2 \ka_0^{3\al-2}} \lesssim \left(\frac{1}{\beta}\right)^{\frac{3}{4}} \left(\frac{\bar{\ka}}{\ka_0}\right)^{\frac{3\al-\beta}{4}} G^{\frac{3}{2}}
\end{equation}

Towards the upper bound in \eqref{thm4:q1}, we have from \eqref{eq:eq1_theta}
\begin{equation} \label{thm4:lathetabd}
    \langle |\la^{\frac{\al}{2}}\te|^2 \rangle = \frac{1}{\ga} \langle (g,\te) \rangle \le \frac{1}{\ga} |g| \langle |\te|^2 \rangle^{1/2} \le \ga \ka_0^{\frac{3\alpha-2}{2}} \bar{\ka}^{\frac{\alpha}{2}} G\langle |\te|^2 \rangle^{1/2} .
\end{equation}
From the upper bound in \eqref{rel:ksigmaeta}, we have $\kappa_\eta^{\alpha} \lesssim \frac{1}{\beta} \left(\frac{\ka_0}{\bar{\ka}}\right)^{\frac{\beta}{3}} \kappa_\sigma^{\alpha}$ and so together with \eqref{thm4:lathetabd}, we can get
\begin{equation} \label{thm4:sim1}
    1 \lesssim \frac{\langle |\la^{\frac{\al}{2}}\te|^2 \rangle^{2/3}}{(\ka_0/\ga)^{2/3}\langle |\te|^2 \rangle} \frac{1}{\beta} \left(\frac{\ka_0}{\bar{\ka}}\right)^{\frac{\beta}{3}} \;,
\end{equation}
and hence
\begin{equation} \label{thm4:al_bd}
 1 \lesssim \frac{(\ga \ka_0^{\frac{3\alpha-2}{2}} \bar{\ka}^{\frac{\alpha}{2}} G)^{2/3}}{(\ka_0/\ga)^{2/3}\langle |\te|^2 \rangle ^ {2/3}} \frac{1}{\beta} \left(\frac{\ka_0}{\bar{\ka}}\right)^{\frac{\beta}{3}}
\end{equation}which gives the upper bound in \eqref{thm4:q1}.

{For the lower bound, we apply} \eqref{eq:eq1_la_nalpha_g} and \eqref{ineq:bterm}, {to obtain}
\begin{align*}
    |\la^{-\frac{\al}{2}}g|^2 &\le \ga |g| \langle |\te|^2 \rangle ^ {1/2} + C \bar{\ka}^{3-\frac{\alpha}{2}}\ka_0^{-\beta} |\la^{-\frac{\alpha}{2}}g| \langle |\la^{\frac{\beta-2}{2}}\te|^2 \rangle \\&\le \ga \bar{\ka}^{\frac{\al}{2}} |\la^{-\frac{\alpha}{2}}g| \langle |\te|^2 \rangle ^ {1/2} + C \bar{\ka}^{3-\frac{\alpha}{2}}\ka_0^{-2} |\la^{-\frac{\alpha}{2}}g| \langle |\te|^2 \rangle.
\end{align*}
{This} can be written as 
\begin{equation*}
    \underbrace{\ga^2 \ka_0^{\frac{3\al-2}{2}} G}_{A_3} \le \underbrace{\ga \bar{\ka}^{\frac{\al}{2}}}_{A_1} \langle |\te|^2 \rangle ^ {1/2} + \underbrace{C \bar{\ka}^{3-\frac{\alpha}{2}}\ka_0^{-2}}_{A_2} \langle |\te|^2 \rangle .
\end{equation*}
Denoting $y = \langle |\te|^2 \rangle ^ {1/2}$, we have
\begin{equation*}
    0 \le A_2 y^2 + A_1 y - A_3.
\end{equation*}
{Hence}
\begin{equation*}
y \ge \frac{-A_1 + \sqrt{A_1^2 + 4A_2A_3}}{2A_2} 
\ge \frac{A_1}{A_2} \;,
\end{equation*}
provided {that} $A_2A_3 \ge 2A_1^2$.
This is equivalent to
\begin{equation*} 
    y^2 = \langle |\te|^2 \rangle \ge \frac{\ga^2}{2C} \frac{\ka_0^{\frac{3\al+2}{2}}}{\bar{\ka}^{3-\frac{\al}{2}}} G, \quad \text{whenever}\quad G \ge
        \frac{2}{C}\left(\frac{\bkap}{\kap_0}\right)^{\frac{3\al-6}2}\overset{\text{def}}{=}\Gb_2.
\end{equation*}
Rearranging gives the {claimed}  lower bound in \eqref{thm4:q1}. 

Finally, {to prove \eqref{thm4:q2}}, we can apply the lower bound from \eqref{thm4:q1} in \eqref{thm4:sim1} to obtain the lower bound {claimed in} \eqref{thm4:q2}. {The claimed upper bound in \eqref{thm4:q2} can be established in a similar way}, by applying the upper bound of \eqref{thm4:q1} into \eqref{thm4:lathetabd}.
\end{proof}

\subsection{Numerical tests of \eqref{G:keta}, \eqref{G:ksigma}}

We had observed directly in subsection \ref{ksigketacompute} that \eqref{etasiglog} holds for the NSE.   This could be expected from the spectrum's adherence to \eqref{withcKr} demonstrated in subsection \ref{computedspectrum}.  It follows that Theorem ~\ref{DFJketaksig} should apply for the NSE.  We examine the numerically computed values of $ \kappa_\eta $ and $ \kappa_\sigma $ to see just how close the lower and upper bounds are to each other.  Note that in our computations $\ka_0=1$.  In Figure~\ref{fig:NSE_G_plots} (left), we plot the quotients $R_\eta /\kappa_\eta $ and $L_\eta /\kappa_\eta $ and find that both bounds are roughly within an order of magnitude of $\kappa_\eta $ over four decades of the Grashof number.   The corresponding quotients for $\ka_\sigma$ in Figure~\ref{fig:NSE_G_plots} (right) show a somewhat wider gap of several orders of magnitude. The fact that the plots are nearly flat confirm that for a turbulent (Kraichnan spectrum) NSE flow both $ \kappa_\eta $ and $\ka_\sigma$ scale as $ G^{1/4}$ up to a log. 

\setlength{\intextsep}{1pt}
\begin{figure}[H]
\centerline{
\includegraphics[width=8.5cm, height=6.5cm]{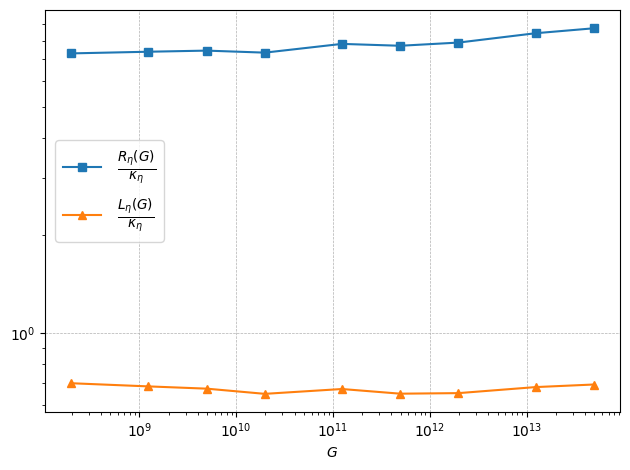} \
\includegraphics[width=8.5cm, height=6.5cm]{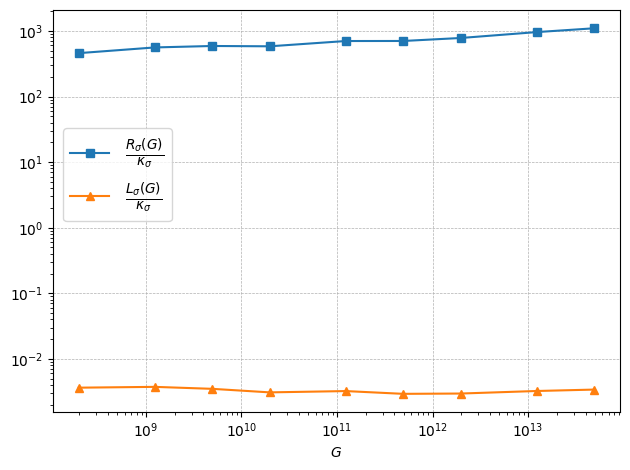}
}
\caption{NSE case, $N=16384$.
Left: Test of \eqref{ketaNSE}. Right: Test of \eqref{ksigNSE} , 
}
\label{fig:NSE_G_plots}
\end{figure}

The corresponding quotients from Theorem ~\ref{main:thm_G_wavenumbers} are plotted for the subcritical SQG in Figure \ref{fig:G_eta_sigma_grid}.   As $\alpha$ decreases toward the critical SQG case, the gap between the upper and lower bounds on both wavenumbers widens by roughly a factor of ten. 

\begin{figure}[H]
\centering

\centerline{
\includegraphics[width=5.5cm, height=5cm]{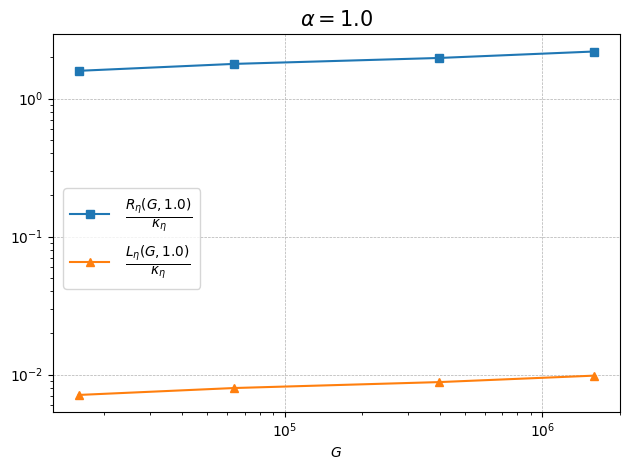}
\
\includegraphics[width=5.5cm, height=5cm]{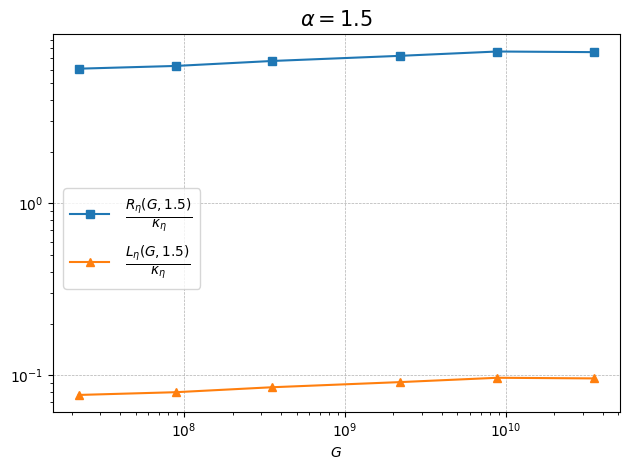}
\
\includegraphics[width=5.5cm, height=5cm]{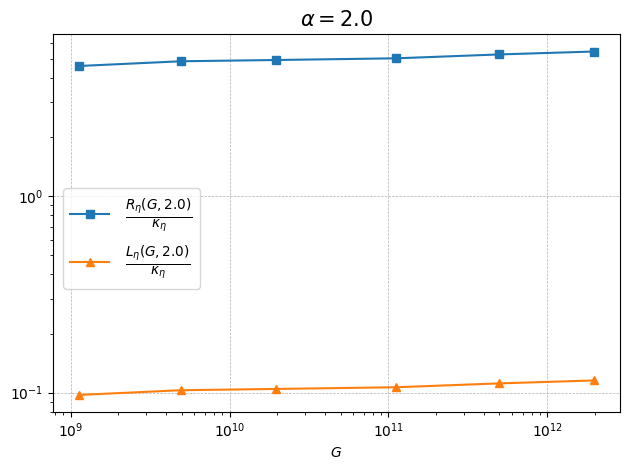}
}

\vspace{5pt} 

\centerline{
\includegraphics[width=5.5cm, height=5cm]{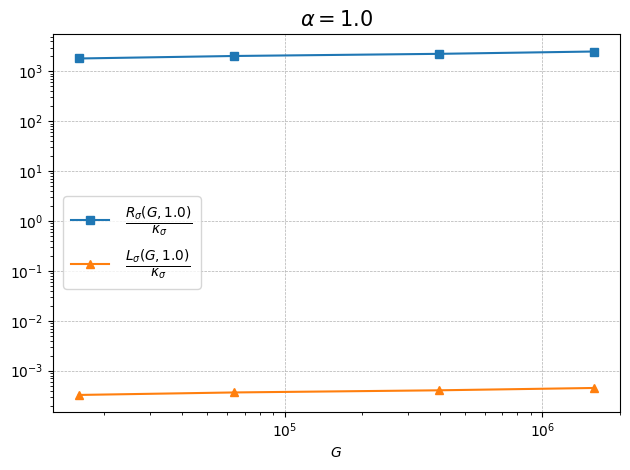}
\
\includegraphics[width=5.5cm, height=5cm]{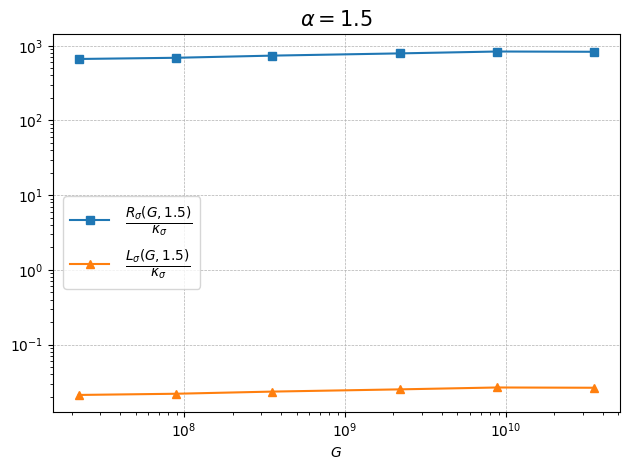}
\
\includegraphics[width=5.5cm, height=5cm]{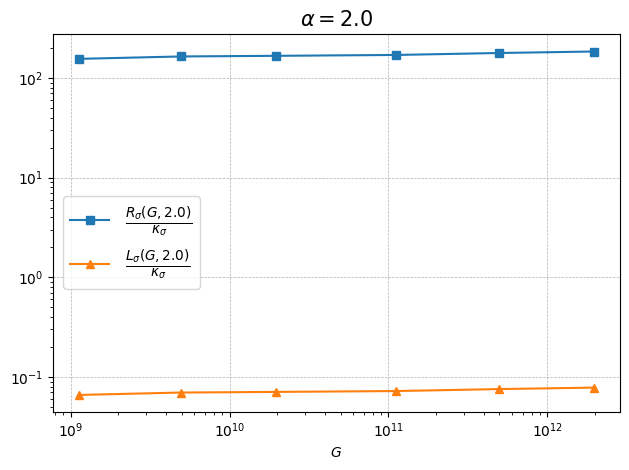}
}

\caption{Comparison of $ \kappa_\eta $ and $ \kappa_\sigma $ with bounds in terms of $ G $ for SQG, $N=16384$. }
\label{fig:G_eta_sigma_grid}
\end{figure}

\begin{figure}[H]
\centering

\centerline{
\includegraphics[width=5.5cm, height=5cm]{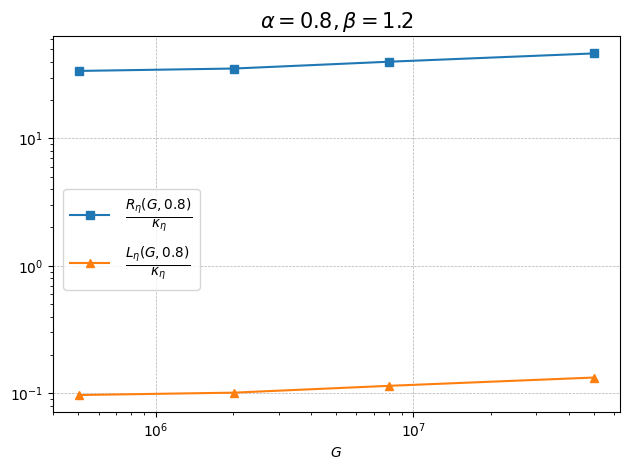}
\
\includegraphics[width=5.5cm, height=5cm]{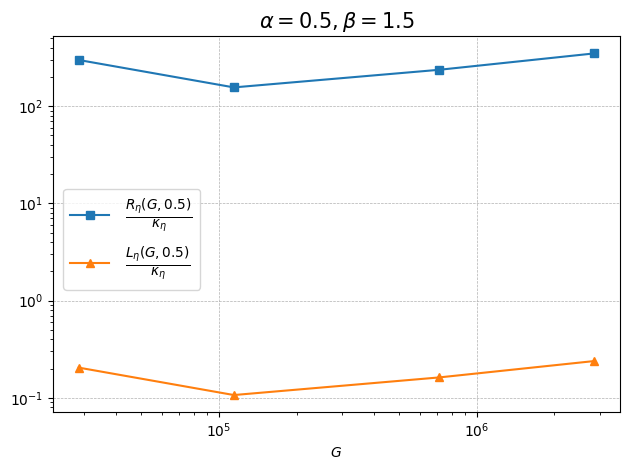}
\
\includegraphics[width=5.5cm, height=5cm]{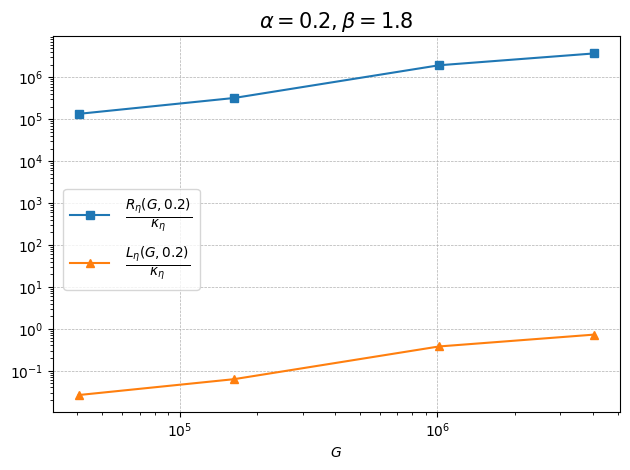}
}

\vspace{5pt} 

\centerline{
\includegraphics[width=5.5cm, height=5cm]{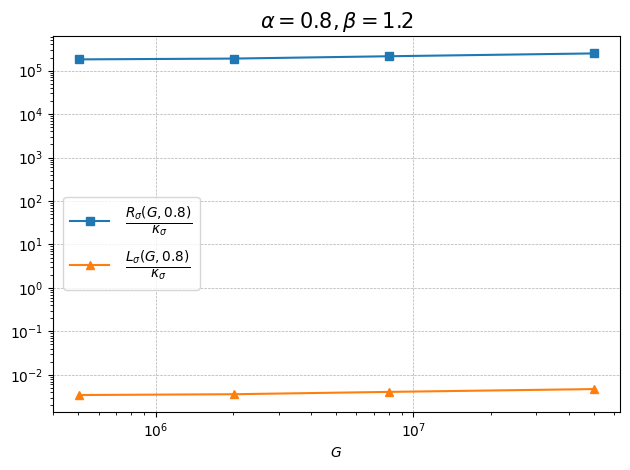}
\
\includegraphics[width=5.5cm, height=5cm]{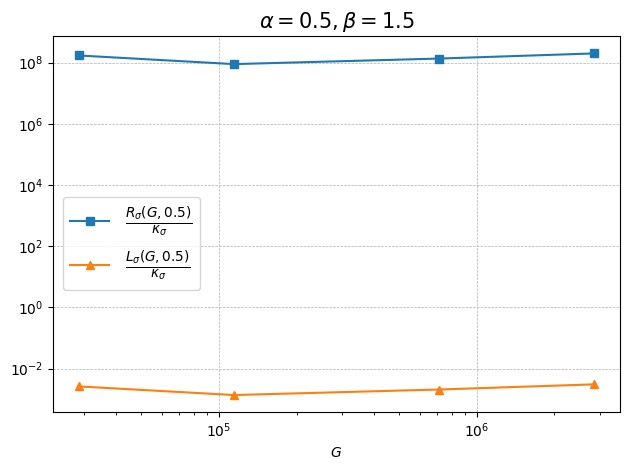}
\
\includegraphics[width=5.5cm, height=5cm]{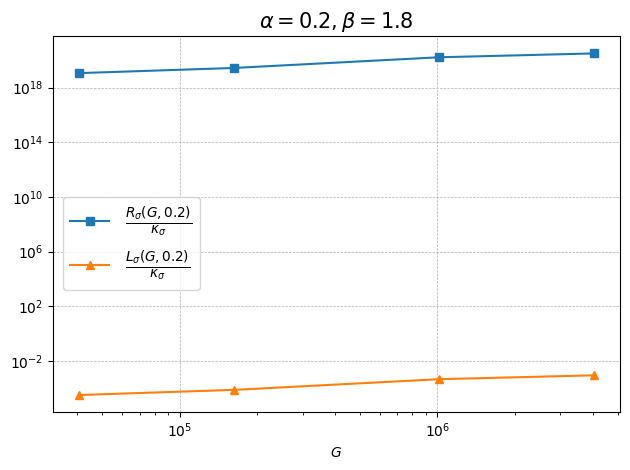}
}

\caption{Comparison of $ \kappa_\eta $ and $ \kappa_\sigma $ with bounds in terms of $ G $ for gSQG, $N=16384$. }
\label{fig:G_keta_ksigma_extended}
\end{figure}

In Figure~\ref{fig:G_keta_ksigma_extended} we consider three samples in the northwest quadrant of the $ (\alpha, \beta) $-plane: one below the critical line $ \beta = 1 + \alpha $, one on the critical line, and one far into the supercritical region.  We see a dramatic widening of the gap between them, particularly for $\ka_\sigma$.  We also note that the quotients increase significantly as $G$ increases, indicating that the wavenumbers grow more slowly than $G^{\frac{1}{2\alpha}}$.  Of course, we had noted a deviation from condition \eqref{rel:ksigmaeta} in the fully nonlinear region, so Theorem~\ref{main:thm_G_wavenumbers} does not apply in this case.

\appendix

\section{Energy Flux}\label{app:energy}

To compute the energy flux \eqref{def:energy:flux}, recall that the energy balance through frequency $\kap$ is (formally) given by
    \begin{align}\notag
        \frac{1}2\frac{d}{dt}|\Lam^{\frac{\be-2}{2}}p_\kap|^2+\gam|\Lam^{\frac{\al+\be-2}{2}}p_\kap|^2=-(u\cdotp\nabla \tht,\Lam^{\be-2}p_\kap)+(g,\Lam^{\be-2}p_\kap),
    \end{align}
where $p_\kap=P_\kap \tht$. Upon recalling that $u=\nabla^\perp\psi$ and 
$\psi=-\ka_0^{-\beta}\Lam^{\be-2}\tht$,
we apply the identities \eqref{eq:identities} and \eqref{psi_orthog} to argue
    \begin{align}
    -(u\cdotp\nabla \tht,\Lam^{\be-2}p_\kap)=\ka_0^\be(\nabla^\perp\psi\cdotp\nabla \tht,P_\kap\psi)=-\ka_0^\be(\nabla^\perp\psi\cdotp\nabla P_\kap\psi,\tht)=-\ka_0^\be(\nabla^\perp Q_\kap\psi\cdotp\nabla P_\kap\psi,\tht)\notag.
    \end{align}
    
Thus, for $q_\kap=Q_\kap\tht$, we have

   \begin{align}
     -(u\cdotp\nabla \tht,\Lam^{\be-2}p_\kap)&=-\ka_0^\be(\nabla^\perp Q_\kap\psi \cdotp\nabla P_\kap\psi,p_\kap)-\ka_0^\be(Q_\kap u\cdotp\nabla P_\kap\psi, q_\kap)\notag
     \\
     &= \ka_0^\be(\nabla^\perp P_\kap\psi \cdotp\nabla Q_\kap\psi,p_\kap)+\ka_0^\be(Q_\kap u\cdotp\nabla  q_\kap, P_\kap\psi)\notag
     \\
     &=-\ka_0^\be(P_\kap u \cdotp\nabla p_\kap,Q_\kap\psi)+\ka_0^\be(Q_\kap u\cdotp\nabla  q_\kap, P_\kap\psi)\notag
     \\
     &=(P_\kap u \cdotp\nabla p_\kap,\Lam^{\be-2}q_\kap)-(Q_\kap u\cdotp\nabla  q_\kap, \Lam^{\be-2}p_\kap)\notag.
    \end{align}
    which yields \eqref{def:energy:flux}.

\section{Littlewood-Paley decomposition}\label{sect:appendix} We give a brief introduction to the Littlewood-Paley decomposition of functions. We state the decomposition for $\mathbb{R}^2$ and point out that it is also valid in the case $\mathbb{T}^2$.
Let $\mathscr{S}(\RR^2)$ denote the space of Schwartz class functions on $\RR^2$ and $\mathscr{S}'(\RR^2)$ denote the space of tempered distributions. We denote by $\hat{f}$ or $\mathcal{F}(f)$, the Fourier transform of $f$, defined by
	\[\hat{f}(\xi)\overset{\text{def}}{=}\int e^{-2\pi i x\cdot \xi}f(x)dx,\quad f\in\mathscr{S}'(\RR^2).\] 
    Recall that for $f,g$, we have
	    \begin{align}\notag
	        (f,g)=(\hat{f},\hat{g}).
	    \end{align}
	The fractional laplacian operator, $\Lam^\si$ is defined as
	    \begin{align}\notag
	        \mathcal{F}(\Lam^\si f)(\xi)=|\xi|^\si\mathcal{F}(f),\quad \si\in\RR.
	    \end{align} 
     For $\si \in \mathbb{R}$, we define the Fourier-based homogeneous and inhomogeneous Sobolev spaces by
	    \begin{align}
	        &\Hdot^\si(\mathbb{R}^2)\overset{\text{def}}{=}\left\{f\in \mathscr{S}(\RR^2):\hat{f}\in L^2_{loc}, \Sob{f}{\Hdot^\si}\overset{\text{def}}{=}|\Lam^\si f|<\infty\right\},\label{def:hom:Sob:norm}\\
	        &H^\si(\mathbb{R}^2)\overset{\text{def}}{=}\left\{f\in \mathscr{S}(\RR^2):\hat{f}\in L^2_{loc}, \Sob{f}{H^\si}\overset{\text{def}}{=}|(I-\Delta)^{\si/2}) f|<\infty\right\}.\label{def:inhom:Sob:norm}
	    \end{align}

We define
	\begin{align*}
	\mathscr{Q}(\RR^2)\overset{\text{def}}{=}\left\{f\in \mathscr{S}(\RR^2): \int f(x)x^{\tau}\, dx=0, \quad \abs{\tau}=0,1,2,\cdots \right\}.
	\end{align*}
	Let $\mathscr{Q}(\RR^2)'$ denote the topological dual of $\mathscr{Q}(\RR^2)$. Then, $\mathscr{Q}(\RR^2)'$ can be identified with the space of tempered distributions modulo the vector space of polynomials on $\mathbb{R}^2$, denoted by $\mathscr{P}$, i.e.
	\begin{align*}
	\mathscr{Q}'(\RR^2)\cong\mathscr{S}(\RR^2)/\mathscr{P}.
	\end{align*}
	Let us denote by ${\Bcal}(r)$, the open ball centered at the origin of radius $r$ and by ${\Acal}(r_{1},r_{2})$, the open annulus centered at the origin with inner and outer radii $r_{1}$ and $r_{2}$. There exist two non-negative radial functions $\chi,\phi\in\mathscr{S}(\RR^2)$ with $\supp\chi\subset{\Bcal}(1)$ and $\supp\phi\subset{\Acal}(2^{-1},2)$ such that for $\phi_j(\xi)\overset{\text{def}}{=}\phi(2^{-j}\xi)$ and $\chi_j(\xi)\overset{\text{def}}{=}\chi(2^{-j}\xi)$, the following conditions are satisfied
	\begin{align*}
	    \begin{cases}
	    \sum_{j\in\ZZ}\phi_j(\xi)=1,\\
	    \chi+\sum_{j\geq0}\phi_j\equiv 1,\,\forall \xi\in\RR^2\setminus\{\mathbf{0}\},\\
	    \supp\phi_i\cap\supp\phi_j=\varnothing,\,\text{if}\,
	    |i-j|\geq2,\\
	    \text{and}\quad\supp\phi_i\cap\supp\chi =\varnothing.
	    \end{cases}
	\end{align*}
    We denote
    \begin{align}\notag            {\Acal}_{j}= {\Acal}(2^{j-1},2^{j+1}),\quad{\Acal}_{\ell,k}= {\Acal}(2^{\ell},2^{k}),\quad  {\Bcal}_j={\Bcal}(2^j).
        \end{align}
    Note that
	    \begin{align}\label{eq:rewrite:supp}
	        \supp\phi_j\subset{\Acal}_j,\quad \supp\chi_j\subset{\Bcal}_j.
	    \end{align}
	we denote by ${\lpj}$ and $S_{j}$, the (homogeneous) Littlewood-Paley dyadic blocks defined as
	\begin{align}\notag
	\mathcal{F}({\lpj}f)=\phi_{j}\mathcal{F}(f), \quad \mathcal{F}(S_{j}f)=\chi_{j}\mathcal{F}(f). 
	\end{align}
    By \eqref{eq:rewrite:supp}, we have
	\begin{align*}
	&\mathcal{F}({\lpj}f)|_{{\Acal}_j^c}=0,\quad
	\mathcal{F}(S_{j}f)|_{{\Bcal}_j^{c}}=0,
	\end{align*}
    For any $f \in \mathscr{S}(\RR^2)$, we have
    \begin{align*}
         f&=S_if+\sum_{j\geq i}\lpj f,\quad i\in\ZZ.
    \end{align*}
	and for any $f\in\mathscr{Q}(\RR^2)'$, we have
	    \begin{align*}
	        f&=\sum_{j\in\ZZ}\lpj f.
	    \end{align*}
We have the following characterization of the Sobolev seminorms
	    \begin{align*}
	       C^{-1}\Sob{f}{\dot{H}^\si}\leq \left(\sum_{j\in \mathbb{Z}}\left(2^{j\si}\nrm{\lpj f}_{L^2}\right)^{2}\right)^{\frac{1}{2}}\leq C\Sob{f}{\dot{H}^\si},
	    \end{align*}
	for some constant $C$ depending only on $\si$. 
    We recall the following inequality which quantifies the relation between the dyadic blocks and the fractional Laplacian operator.
	\begin{lemma}
	    [Bernstein inequalities]\label{T:Bernstein}
		Let $\si\in\RR$ and $1\le p \le q\le \infty$. Then
		\begin{align*}
	C^{-1}2^{\si j}\nrm{{\lpj}f}_{L^q}\le \nrm{\Lam^{\si}{\lpj}f}_{L^q}\le C 2^{\si j+2j(\frac{1}{p}-\frac{1}{q})}\nrm{{\lpj}f}_{L^p},
		\end{align*}
	 where $C>0$ is a constant that depends on $p,q$ and $\si$.
	\end{lemma}
    
\begin{proof}[Proof of Lemma \ref{JKM22lemma}]

    We define the sum
    \[ {\mathcal{L}}_{s,\ell}(f_1,f_2,f_3)\overset{\text{def}}{=}\sum_{(u,v)\in \mathbb{Z}^2}m_{s,\ell}(u,v)\widehat{f_1}(u)\widehat{f_2}(v-u)\overline{\widehat{f_3}(v)},
    \]
where 
\[m_{s,\ell}(u,v)\overset{\text{def}}{=}|v|^{-s}v-|v-u|^{-s}(v-u),\]
and observe that 
\[\lbn [\la^{-s}\nabla,f_1]f_2,f_3\rbn=\mathcal{L}_{s,\ell}(f_1,f_2,f_3).\]
Now let
    \begin{align}\label{def:A}
    \mathbf{A}(u,v)(\tau)\overset{\text{def}}{=}\tau v+(1-\tau)(v-u).
    \end{align}
Henceforth, we suppress the dependence of $\mathbf{A}$ on $u,v$. Observe that
    \begin{align}\label{E:meanvalue1}
        |m_{s,\ell}(u,v)|&=\abs{\int_{0}^{1}\frac{d}{{d\tau}}\left(\abs{\mathbf{A}(\tau)}^{-s}\mathbf{A}(\tau)\right)d\tau}\notag\\
        &=\abs{\int_{0}^{1}\left(-s\abs{\mathbf{A}(\tau)}^{-s-2}(\mathbf{A}(\tau)\cdot u)\mathbf{A}(\tau)+\abs{\mathbf{A}(\tau)}^{-s}u\right)d\tau}\notag\\
    &\le C\abs{u}\int_{0}^{1}\abs{\mathbf{A}(\tau)}^{-s}d\tau,
    \end{align}
where the fact $s\in(0,1)$ is invoked to obtain the last inequality. Since $\text{supp} \ \widehat{f_2}\subset \mathcal{A}_i$ and $\text{supp} \ \widehat{f_3}\subset \mathcal{A}_j$, we can assume that $\text{supp} \ \widehat{f_1}\subset \mathcal{B}_{i+k+2}$. We consider two cases: 
\subsubsection*{Case 1: $\text{supp} \ \widehat{f_1}\subset \mathcal{B}_{i-3}$}
For $u\in \mathcal{B}_{i-3}$, we have
\[|\mathbf{A}(\tau)|\ge |v-u|-\tau |u|\ge 2^{i-1}-2^{i-3}=3(2^{i-3})\ge \frac{3}{16}|v-u|.\]
For $\varphi\overset{\text{def}}{=}\frac{v-u}{|u|}$ and $\vartheta\overset{\text{def}}{=}\frac{u}{|u|}$, we therefore obtain 
\begin{align}\label{BoundEstimate1}
    \int_{0}^{1}\abs{\mathbf{A}(\tau)}^{-s}d\tau \le C |v-u|^{-\rho}|u|^{\rho-s}\int_{0}^{1}\frac{1}{|\varphi+\tau \vartheta |^{s-\rho}}d\tau
    \le C |v-u|^{-\rho}|u|^{\rho-s},
\end{align}
where we invoked Lemma 3.2 in \cite{JKM2022} for the last inequality. 
\subsubsection*{Case 2: $\text{supp} \ \widehat{f_1}\subset \mathcal{A}_{i-3,i+j+k}$}. Once again invoking Lemma 3.2 in \cite{JKM2022}, we have 
\begin{align}\label{BoundEstimate2}
    \int_{0}^{1}\abs{\mathbf{A}(\tau)}^{-s}d\tau \le C|u|^{-s}\le C_k |v-u|^{-\rho}|u|^{\rho-s}.
\end{align}
Using \eqref{BoundEstimate1} and \eqref{BoundEstimate2} in \eqref{E:meanvalue1}, we obtain
\begin{align*}
    |\mathcal{L}_{s,\ell}(f_1,f_2,f_3)|\le C \sum_{(u,v)\in \mathbb{Z}^2}|\widehat{\Lam^{1+\rho-s}f_1}(u)||\widehat{\Lam^{-\rho}f_2}(v-u)|{\widehat{f_3}(v)}|.
\end{align*}

Finally, applying the Cauchy-Schwarz inequality and Young's inequality gives us the claimed result.

\end{proof}

\bibliography{reference6Oct2025}
\bibliographystyle{plain} 

\end{document}